\documentclass{osa-article}
\usepackage{amsmath}
\usepackage{mathtools}
\usepackage{tabularx}
\usepackage{graphicx}
\usepackage[labelformat=simple]{subcaption}

\usepackage{color}
\newtheorem{theorem}{Theorem}

\newtheorem{proof}{Proof}
\newtheorem{proposition}{Proposition}
\DeclareMathOperator\erf{erf}
\newcolumntype{Z}{>{\centering\arraybackslash}X}

\journal{boe}

\begin{document}

\title{Optical Wireless Cochlear Implants}

\author{Stylianos E. Trevlakis,\authormark{1,2} Alexandros-Apostolos A. Boulogeorgos,\authormark{1, 3} Paschalis C. Sofotasios,\authormark{2,4,*} Sami Muhaidat,\authormark{2} and George K. Karagiannidis\authormark{1}}

\address{\authormark{1}Department of Electrical and Computer Engineering, Aristotle University of Thessaloniki, GR-54124 Thessaloniki, Greece\\
\authormark{2}Department of Electrical and Computer Engineering, Khalifa University of Science and Technology, UAE-127788 Abu Dhabi, United Arab Emirates\\
\authormark{3}Department of Digital Systems, University of Piraeus, GR-18534 Piraeus, Greece\\
\authormark{4}Department of Electronics and Communications Engineering, Tampere University of Technology, FI-33720 Tampere, Finland}

\email{\authormark{*}paschalis.sofotasios@ku.ac.ae} 

\begin{abstract}
In the present contribution, we introduce a wireless optical communication-based system architecture which is shown to significantly improve the reliability and the spectral and power efficiency of the transcutaneous link in cochlear implants (CIs). 
We refer to the proposed system as optical wireless cochlear implant (OWCI).
In order to provide a quantified understanding of its design parameters, we establish a theoretical framework that takes into account the channel particularities, the integration area of the internal unit, the transceivers misalignment, and the characteristics of the optical units. 
To this end, we derive explicit expressions for the corresponding average signal-to-noise-ratio, outage probability, ergodic spectral efficiency and capacity of the transcutaneous optical link (TOL).  
These expressions are subsequently used to assess the dependence of the TOL's communication quality on the transceivers design parameters and the corresponding channels characteristics. 
The offered analytic results are corroborated with respective results from Monte Carlo simulations. 
Our findings reveal that OWCI is a particularly promising architecture that drastically increases the reliability and effectiveness of the CI TOL,  whilst it requires considerably lower transmit power compared to the corresponding widely-used radio frequency (RF) solution. 
\end{abstract}

\section{Introduction}
During the past decades, medical implants have been advocated as an effective solution to numerous health issues due  to the quality of life improvements they can provide. 
One of the most successful application of such devices is cochlear implants (CIs), which have restored partial hearing to more than $350,000$ people worldwide, half of whom are pediatric users who ultimately develop nearly normal language capability~\cite{A:The_modern_cochlear_implant,A:Cochlear_Implants_System_design_integration_and_evaluation}.
A typical CI consists of an out-of-body unit, which uses a microphone to capture the sound, converts it into  a radio frequency (RF) signal and wirelessly transmits it to an in-body unit.
Conventional CIs exploit near-field magnetic communication technologies and typically operate in low RF frequencies, from $5$ MHz to $49$ MHz, while their transmit power is in the order of tens of $\mathrm{mW}$~\cite{A:Cochlear_Implants_System_design_integration_and_evaluation,A:MED-EL_CI,A:NFCI,A:WPT_strategies_for_implantable_bioelectronics}. 
However, the main disadvantage of this technology is that it cannot support high data rates which are required for neural prosthesis applications in order to achieve similar performance to that of human organs, such as the cochlea, under reasonable transmit power constraints~\cite{8052089,A:IR_neural_stimulation,A:Early_History_and_Challenges_of_Implantable_electronics}.
By also  taking into account the interference from other sources operating in the same frequency band, RF transmission is largely  rendered  a mediocre solution~\cite{islam2016review,liu2012optical,pinski2002interference}. 

Overcoming the above constraints in CIs  calls for investigations on  the feasibility of transcutaneous  wireless links that operate  in non-standardized frequency bands. 
In light of this  and due to the increased bandwidth availability, the partial transparency of skin at infrared wavelengths and the remarkably  high immunity to external interference,  the use of optical wireless communications (OWCs) has been recently introduced as an attractive alternative solution to the conventional approach (see~\cite{A:Emerging_OWC_Advances_and_Challenges,ghassemlooy2017visible,OWC_vs_RF_survey} and the references therein). 
The feasibility of OWCs for the transcutaneous optical link (TOL) has been experimentally validated in numerous contributions~\cite{abita2004transdermal,ackermann2006design,ackermann2008designing,gil2012feasibility,liu2012optical,liu2013system,liu2014vivo,liu2015bidirectional,okamoto2005development}.
Specifically, TOLs were used  in~\cite{abita2004transdermal} to establish transdermal high data rate communications between the internal and external units of a medical system. 
Additionally,  the TOL's primary design parameters and their interaction were quantified  in~\cite{ackermann2006design}, 
 whilst the authors in~\cite{ackermann2008designing}  reported the tradeoffs related to the design of OWC based CI.
In the same context, the authors in~\cite{liu2013system}  evaluated the performance of a TOL utilized for clinical neural recording purposes in terms of tissue thickness, data rate and transmit power. 
Furthermore, the characteristics of the receiver were investigated with regard to its size and the signal-to-noise-ratio~(SNR) maximization. 
In~\cite{liu2012optical}, a low power high data rate  optical wireless link for implantable transmission was presented and evaluated in terms of power consumption and bit-error-rate (BER), for a predetermined misalignment tolerance. 
Likewise,  the authors in~\cite{liu2014vivo} verified in-vivo the feasibility of TOL, proving that high data rates - even in the order of $100$ Mbps - can be delivered with a BER of $2\times 10^{-7}$  in the presence of misalignment fading. Nevertheless,  these  high data rates have been achieved at the cost of high power consumption~\cite{6679681}, which reached $2.1\text{ }mW$.
In contrast to~\cite{liu2014vivo},  the authors in~\cite{gil2012feasibility} presented novel experimental results of  direct and retroreflection transcutaneous link configurations and their findings verified the feasibility of highly effective and robust low-power consumption TOLs. 
Finally,  a bidirectional transcutaneous optical telemetric data link was reported  in~\cite{liu2015bidirectional}, whereas   a bidirectional transcutaneous optical data transmission system for artificial hearts,  allowing long-distance data communication with low electric power consumption, was described in~\cite{okamoto2005development}. 

\begin{table}[]
    \centering
    \caption{\textcolor{black}{Comparison between OWCIs and RFCIs.}}
    \resizebox{0.88\columnwidth}{!}{%
    \def\arraystretch{1.3}
    \begin{tabularx}{\textwidth}{|*{2}{Z|}}
        \hline
        \textbf{OWCIs} & \textbf{RFCIs}\\
        \hline
        \hline
        \textbf{Unprecedented increase} & Relatively limited \\
        \textbf{ in bandwidth} & bandwidth \\
        \hline
        \textbf{High achievable data rate} & Mediocre data rate \\
        \hline
        \textbf{Low interference: Solar and} & \\
        \textbf{ambient light are the main sources} & Very high interference from other \\
        \textbf{of interference and can be} & electronic and electrical appliances \\
        \textbf{easily mitigated or even cancelled} & \\
        \hline
        \textbf{Safer for the human health} & Of questionable safety \\
        \textbf{due to the low transmission power} & concerning the human body \\
        \hline
        \textbf{Low power demands} & High power demands \\
        \hline
        \multicolumn{2}{|c|}{Comparable cost} \\
        \hline
        \textbf{Mature technology that can exploit} & \\
        \textbf{the particularities of novel materials} & Mature technology \\
        \textbf{(e.g. graphene) to achieve improved} & that allows compact designs \\
        \textbf{features in the same integration area} & \\
        \hline
        Existing design guidelines from several & \textbf{Matured} \\
        standards (i.e. IEEE Std 1073.3.2-2000, & \textbf{standardization} \\
        IEC Laser Safety Standard, IrDA, & \textbf{(i.e. IEEE 802.15.4,} \\
        JEITA\! Cp-1223,\! European\! COST 1101,\! etc.) & \textbf{IEEE Std 1073.3.2-2000, etc.)} \\
        \hline
        Stringent alignment requirements & \textbf{Susceptible misalignment} \\
        \hline
    \end{tabularx}}
    \label{tab:OWC_vs_RF}
\end{table}

To the best of the authors' knowledge, the use of OWCs for establishing TOLs in CIs has not been reported in the open technical literature. 
Motivated by this, in the present work, we propose a novel system architecture which by employing OWCs improves the reliability as well as the spectral and power efficiency of the TOL. 
We refer to this architecture as optical wireless cochlear implant (OWCI). 
\textcolor{black}{A list of the advantages and disadvantages of the OWCI and RF cochlear implant (RFCI) is depicted in Table~\ref{tab:OWC_vs_RF}.}
It is noted that OWCI is in line with the advances in CIs, since the use of optics for the stimulation of the acoustic nerve has been recently proposed~\cite{A:IR_neural_stimulation,A:GaN_based_micro_LED_arrays_on_flexible_substrates_for_optical_cochlear_implants,kallweit2016optoacoustic,schultz2014optical,RICHTER201472,duke2009combined}. 
\textcolor{black}{In addition, we evaluate the feasibility and the capabilities of the proposed system whilst we provide design guidelines for the OWCI. 
Specifically, the technical contribution of this paper is summarized below:
\begin{itemize}
\item We establish an appropriate system model for the TOL, which includes all  different design parameters and their corresponding  interactions.
These parameters include the thickness of the skin through which the light is transmitted, the size of the integration area of the optics, the degree of transmitter (TX) and receiver (RX) misalignment, the efficiency of the optical system, and the emitter power.
\item We derive a comprehensive  analytical framework that quantifies and evaluates  the feasibility and effectiveness of the OWC link in the presence of misalignment fading. 
In this context, we provide novel closed-form expressions for the instantaneous and average SNR, which quantifies the received signal quality. 
Additionally, we evaluate the outage performance and the capacity of the optical link by deriving tractable analytic expressions for the outage probability and the spectral efficiency. 
These expressions take into account the technical characteristics of the link as well as the transcutaneous medium particularities; hence, they provide  meaningful insights into the behavior of the considered set up, which can be used as  design guidelines of such systems. 
Finally, we deduce novel expressions and simple bounds/approximations for the evaluation of the OWCI capacity. 
Interestingly, our findings reveal the superior reliability and effectiveness of  OWCI compared to the baseline CI RF solution.
\end{itemize} }

The remainder of this paper is organized as follows:  
The system model of the OWCI is described in Section~\ref{S:Sm}.
In Section~\ref{S:Pa}, we provide the analytical framework for quantifying the performance of the OWCI  by deriving novel analytic expressions for the instantaneous and average SNR, the outage probability, as well as the spectral efficiency and the channel capacity.
Respective numerical and simulation results, which illustrate the performance of the OWCI along with  useful related discussions are provided in Section~\ref{S:Nr}.
Finally, closing remarks and a summary of  the main findings of this contribution are presented in Section~\ref{S:C}.

\textit{Notations:} Unless stated otherwise, in this paper, $|\cdot|$ denotes  absolute value,  $\exp(\cdot)$  represents the exponential function, while  $\log_2(\cdot)$ and $\ln(\cdot)$ stand for the binary logarithm  and the natural logarithm, respectively.
In addition, $P\left(\mathcal{A}\right)$ denotes the probability of the event $\mathcal{A}$, whereas  $\mathrm{erf}(\cdot)$  denotes  the error function
and $\Phi(\cdot , \cdot , \cdot)$ represents the Lerch $\Phi$ function~\cite{B:Gra_Ryz_Book}.

\section{System model}
\label{S:Sm}
As illustrated in~Fig.~\ref{fig:SM}, the main  parts  of the  OWCI are the external unit, the propagation medium (skin) and the internal unit~\cite{dorman2004design,A:Cochlear_Implants_A_remarkable_past_and_a_briliant_future}.
The external unit consists of a microphone that captures the sound, followed by the sound processor responsible for the digitization and compression of the captured sound into coded signals.
The coded signal is then forwarded to the transmitter (TX), which conveys the data to the RX in the internal unit.
The internal unit consists of the RX which is embedded in the skull, a digital signal processing (DSP) unit, a stimulation (STM) unit and an electrode array implanted in the cochlea.
The DSP and STM units operate together in order to modulate the received signal into pulses of electrical current that are capable of stimulating the auditory nerve. 
Finally, the electrode array delivers the signal to  the auditory nerve, where it is  interpreted as sound by the brain.  

\begin{figure}[htbp]
\centering\includegraphics[width=1\linewidth,trim=0 0 0 0,clip=false]{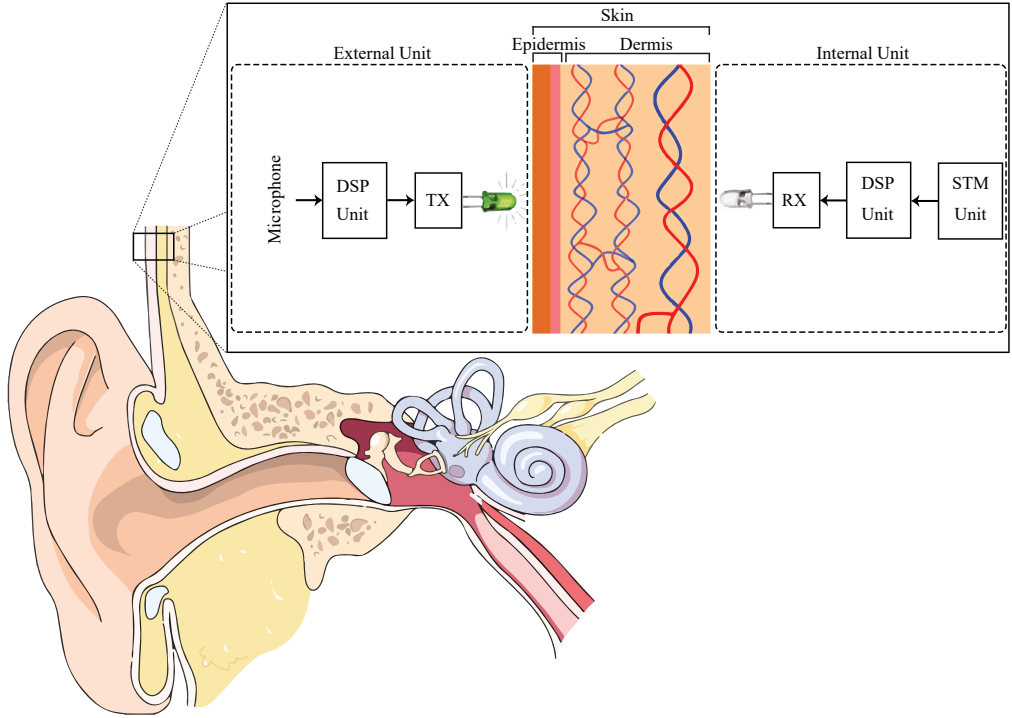}
\caption{System architecture of OWCI. In this figure `DSP unit' and `STM unit' denote the digital signal processing and stimulation units, respectively.}\label{fig:SM}
\end{figure}

We assume that the transmitted signal, $x$, is conveyed  over the wireless channel, $h$, with additive noise $n$. 
Therefore, the baseband equivalent received signal can be expressed~as~\cite{C:OWC_SystemModel_Capacity_and_Coding,zedini2015multihop,popoola2009bpsk}
\begin{align}
y = R h x + n
\label{Eq:Received_signal}
\end{align}
where $R$ denotes  the responsivity of the RX's photodiode,  which is given by 
\begin{align}
R=\eta\frac{q}{p v}.
\label{Eq:Responsivity} 
\end{align}
In~\eqref{Eq:Responsivity}, $\eta$ represents the quantum efficiency of the photodiode, $q$ is the electron charge, $v$ denotes the photons' frequency and $p$ is the Planck's constant. 

Likewise, the channel, $h$, can be expressed as~\cite{gil2012feasibility}
\begin{align}
h = h_l h_p 
\end{align}
where $h_l$ and $h_p$ denote, respectively, the deterministic channel coefficient, due to the propagation loss, and the stochastic process that models the geometric spread, due to the misalignment fading. 
\textcolor{black}{
Also, we assume a Gaussian spatial intensity profile of the beam waste on the RX plane at distance $\delta$ from the TX, $w_{\rm \delta}$, and a circular aperture of radius $\beta$. 
As a result, the stochastic term of the channel coefficient, $h_p$, represents the fraction of the collected power due to geometric spread with radial displacement $r$ from the origin of the detector.
Moreover, by assuming that the elevation and the horizontal displacement (sway) follow independent and identical Gaussian distributions, as in~\cite{A:Arnon2003Effects}, we observe that the radial displacement at the RX follows a Rayleigh distribution.
}

To this effect, the deterministic term of the channel coefficient can be evaluated as~\cite[eq.~10.1]{B:Transdermal_optical_communications}
\begin{align}
h_l=\exp\left(-\frac{1}{2}\alpha(\lambda) \delta\right) 
\end{align}
where $\alpha(\lambda)$  represents the skin attenuation coefficient at  wavelength $\lambda$ and  $\delta$ denotes the total dermis thickness, which is approximately equivalent to the TX-RX distance.
The skin attenuation coefficient can be derived from~\cite{bashkatov2011optical,A:graaff1993Opticalpropertiesofhumandermisinvitroandinvivo,A:Chang1996Effectsofcompressiononsofttissueopticalproperties,A:Simpson1998Near-infraredopticalproperties,A:Du2001Opticalpropertiesofporcineskindermis,A:Troy2001Opticalpropertiesofhumanskin,A:Bashkatov2005Opticalpropertiesofhumanskin}. 
It is worth noting that according to~\cite{B:Transdermal_optical_communications}, the term skin refers to the complex biological structure composed by three essential layers, namely stratum corneum, epidermis, and dermis.
Furthermore, as in~\cite{liu2012optical}, the TX and the RX are in contact with the outer (epidermal) and inner (adipose) side of the skin.
As a result, the distance between  TX and   RX can be approximated by the skin thickness.    

\color{black}
In the same context, the noise component can be expressed~as
\begin{align}
n =  n_{\rm s} + n_{\rm th}
\end{align}
where $n_{\rm s}$, and $n_{\rm th}$ represent the shot noise and the thermal noise, respectively, which can be modelled as zero-mean Gaussian processes~\cite{ghassemlooy2017visible}.
    
In more detail, the shot noise can be expressed~as
\begin{align}
n_{\rm s} =  n_b + n_{\rm DC}
\end{align}
where $n_b$, and $n_{\rm DC}$ represent  the background shot noise and dark current shot noise, respectively, which can be also modeled as  a zero-mean Gaussian processes with~variances 
\begin{align}
\label{Eq:s_b2}
\sigma_b^2  = 2 q R B P_b 
\end{align} 
and
\begin{align}
\label{Eq:s_b2b}
\sigma_{\rm DC}^2 = 2 q  B I_{\rm DC}
\end{align} 
respectively~\cite{gil2012feasibility,popoola2009bpsk}. It is noted that the term $B$ in \eqref{Eq:s_b2} and \eqref{Eq:s_b2b}, denotes the communication bandwidth, $R$ is the responsivity of the detector,  $P_b$ is the background optical power, and $I_{\rm DC}$ is the intensity of the dark current which  is generated by the photodetector in the absence of background light, and stems from thermally generated electron–hole pairs.
    
Finally, $n_{\rm th}$ represents the thermal noise, which is caused by thermal fluctuations of the electric carriers in the receiver circuit, and can be modelled as a zero-mean Gaussian process with variance $\sigma_{\rm th}^2$. As a result, since  $n_b$,  $n_{\rm DC}$, and $n_{\rm th}$ are zero-mean Gaussian processes~\cite{gil2012feasibility,ghassemlooy2017visible}, $n$ also follows a zero-mean Gaussian distribution with variance 
\begin{align}
\label{Eq:s2}
\sigma^2 = \sigma_{\rm b}^2 + \sigma_{\rm DC}^2 + \sigma_{\rm th}^2.
\end{align}
\color{black}

\section{Performance analysis}
\label{S:Pa}
In this section, we provide the mathematical framework that quantifies the performance of the proposed system. 
In this context, we derive closed-form expressions for the instantaneous and average SNR, outage probability and ergodic spectral efficiency and capacity. 
Capitalizing on them, we also derive simple and tight   lower bounds for the ergodic capacity that provide useful insights. 
Hence, these expressions can be used to validate the feasibility of the system as well as to introduce design guidelines for its effective~utilization.

\subsection{Evaluation of the average SNR}
Based on~\eqref{Eq:Received_signal}-\eqref{Eq:s2}, the instantaneous SNR can be obtained~as~\cite{7862126}
\begin{align}
\gamma = \frac{R^2 \exp\left(-\alpha(\lambda) \delta\right) h_p^2 P_s}{2 q R B P_b + 2 q  B I_{\rm DC} + \sigma_{\rm th}^2}
\label{Eq:SNR}
\end{align}
or equivalently
\begin{align}
\gamma = \frac{R^2 \exp\left(-\alpha(\lambda) \delta\right) h_p^2 \tilde{P}_s}{2 q R P_b + 2 q  I_{\rm DC} + N_0}
\label{Eq:SNR_eq}
\end{align}
where $P_s$ denotes the average optical power of the transmitted~signal, whereas $\tilde{P}_s$ and $N_0$  represent the signal and noise optical power spectral density (PSD), respectively.

The  average SNR in the considered set up is derived in the following theorem. 
\begin{theorem}
The average SNR in the considered system can be expressed as
\begin{align}
\tilde{\gamma}=\frac{R^2 \exp\left(-\alpha(\lambda) \delta\right) \tilde{P}_s}{2 q R  P_b + 2 q   I_{\rm DC} +N_0} \frac{\xi A_0^2}{\xi+2}
\label{Eq:AverageSNR}
\end{align}
where ~\cite{A:Outage_Capacity_for_FSO_with_pointing_errors}
\begin{align}
A_0 &= [\erf\left(\upsilon\right)]^2
\label{Eq:A_0}
\end{align}
denotes the    fraction of the collected power in case of zero radial displacement, 
with   
\begin{align}
\upsilon = \frac{\sqrt{\pi}\beta}{\sqrt{2}w_{\rm \delta}}.
\label{Eq:v}
\end{align}
In the above equations, $\beta$ and $w_{\rm \delta}$ denote, respectively, the radius of the RX's circular aperture and the beam waste (radius calculated at $e^{-2}$) on the RX plane at distance $\delta$ from the TX. 
Moreover,   $\xi$ is the square ratio of the equivalent beam radius, $w_{\rm eq}$, and the pointing error displacement standard deviation at the RX, namely 
\begin{align}
\xi=\frac{w_{\rm eq}^2}{4\sigma_s^2}
\label{Eq:xi}
\end{align}
with $\sigma_s^2$ denoting  the pointing error displacement (jitter) variance  at the RX, whereas
\begin{align}
w_{\rm eq}^2 &= w_{\rm \delta}^2\frac{\sqrt{\pi}\erf\left(\upsilon\right)}{2\upsilon\exp\left(-\upsilon^2\right)}.
\label{Eq:wZeq}
\end{align}
\end{theorem}
\begin{proof}
The proof is provided in Appendix~A. 
\end{proof}

\textcolor{black}{From~\eqref{Eq:AverageSNR}, it is evident that the average SNR depends on the transmission PSD, the skin particularities, namely skin attenuation and thickness, the RX's characteristics, and the intensity of the misalignment fading. 
It is also noted  that since the skin attenuation is a function of the wavelength, the average SNR is also a function of the wavelength. }

\subsection{Evaluation of the ergodic spectral efficiency}
In this section, we quantify the OWCI capability in terms of spectral efficiency.
In this context,  a novel closed-form expression for the ergodic spectral efficiency of the proposed system is derived in the following theorem.

\begin{theorem}
The ergodic spectral efficiency of the considered set up can be expressed as
\begin{align}
C = \frac{1}{2} \log_2\left(1+ \mathcal{B}(\lambda) A_0^2\right) - \frac{1}{2} \frac{A_0^{2} \mathcal{B}(\lambda)}{\ln(2)} \Phi\left(-A_0^2 \mathcal{B}(\lambda), 1, 1+\frac{\xi}{2}\right) 
\label{Eq:ergodic_cap}
\end{align}
where
\begin{align}
\mathcal{B}(\lambda)=\frac{\psi R^2 \exp\left(-\alpha(\lambda) \delta\right)\tilde{P}_s}{2 q R P_b + 2 q  I_{\rm DC} + N_0}.
\label{Eq:B}
\end{align}
\end{theorem}
\begin{proof}
The proof is provided in Appendix~B. 
\end{proof}
 
It is recalled that  the Shannon ergodic spectral efficiency in $C = \tfrac{1}{2} \mathbb{E}\left[\log_2\left(1 + \psi \gamma\right)\right]$ is a valid   exact expression for deriving the ergodic  spectral efficiency for the heterodyne detection technique, while it acts as a lower bound for  the IM/DD scheme~\cite[eq. (26)]{A:OnTheCapacityOfFSOIntensityChannels} and~\cite[eq. (7.43)]{B:Advanced_OWC_Systems}. 
Hence,   the derived ergodic spectral efficiency in~\eqref{Eq:ergodic_cap} is correspondingly an exact solution for the heterodyne detection technique and   a lower bound for the IM/DD technique.

Furthermore, it is noticed that  \eqref{Eq:ergodic_cap} is an  exact closed-form expression that can be straightforwardly computed using popular software packages. 
Based on this,  we  can also derive a simple  lower bound for the ergodic spectral efficiency that provide useful insights on the impact of the involved parameters. 

\begin{proposition}
The ergodic spectral efficiency can be lower bounded as follows: 
\begin{align}
C > \frac{1}{2} \log_2\left(1+ \mathcal{B}(\lambda) A_0^2\right) -  \frac{1}{\xi \ln(2)} 
\label{Eq:ergodic_cap_b}
\end{align}
\end{proposition}

\begin{proof}
The proof is provided in Appendix C. 
\end{proof}

It is noted here that the above lower bound  is particularly tight; as a result, equation  \eqref{Eq:ergodic_cap_b} can also considered as a highly accurate bounded approximation of the ergodic spectral efficiency. 

\textcolor{black}{Based on the above and with the aid of~\eqref{Eq:AverageSNR}, it follows that \eqref{Eq:ergodic_cap} can be equivalently rewritten~as
\begin{align}
C =  \frac{1}{2} \log_2\left(1+ \psi \frac{\xi+2}{\xi}\tilde{\gamma}\right) - \frac{\psi}{2 \ln(2)} \frac{{\xi+2}}{{\xi}} \tilde{\gamma} \Phi\left(-\psi \frac{\xi+2}{\xi}\tilde{\gamma}, 1, 1+\frac{\xi}{2}\right)
\label{Eq:ergodic_cap2}
\end{align}
which with the aid of Proposition 1 it can be lower bounded as follows: 
\begin{align}
C =  \frac{1}{2} \log_2\left(1+ \psi \frac{\xi+2}{\xi}\tilde{\gamma}\right) - \frac{1}{\xi \ln(2)}.  
\label{Eq:ergodic_cap2_b}
\end{align}
It is noticed in~\eqref{Eq:ergodic_cap2} that the ergodic spectral efficiency depends on the level of misalignment, modeled by the term $\xi$, and the average SNR.}

\subsection{Evaluation of the ergodic capacity}
Based on the presented channel model, the optical transcutaneous link is wavelength selective. 
Also, the bandwidth over which the channel transfer function remains virtually constant is hereby denoted as $\Delta f$.
Thus, the capacity can be obtained by dividing the total bandwidth, $B$, into $K$ narrow sub-bands and  then summing up  the individual capacities~\cite{oppermann2005uwb}. 
The $i^{th}$ sub-band is centered around the wavelength $\lambda_i$, with $i\in[1, K]$, and it has width $\Delta f$. 
If the subband width is sufficiently small, the sub-channel appears as wavelength non-selective. 
As a consequence, by assuming full CSI knowledge at both the TX and the RX, the resulting capacity in $\mathrm{bits/s}$ can be expressed~as
\begin{align}
C_{\rm B} = \frac{1}{2} \sum_{\rm i=1}^{K}\Delta f \left(\log_2\left(1+ \mathcal{B}(\lambda_i) A_0^2\right) - \frac{A_0^{2} \mathcal{B}(\lambda_i)}{\ln(2)} \Phi\left(-A_0^2 \mathcal{B}(\lambda_i), 1, 1+\frac{\xi}{2}\right)\right).
\label{Eq:Capacity}
\end{align}
Of note,  for  the special case in which $B\leq \Delta f$, equation \eqref{Eq:Capacity} becomes
\begin{align}
C_B= \frac{1}{2} B \log_2\left(1+ \mathcal{B(\lambda)} A_0^2\right) - \frac{1}{2} B \frac{A_0^{2} \mathcal{B(\lambda)}}{\ln(2)} \Phi\left(-A_0^2 \mathcal{B(\lambda)}, 1, 1+\frac{\xi}{2}\right)
\label{Eq:Capacity2}
\end{align}   
which is   a common case since the values of  $\Delta f$ are in the order of $\mathrm{GHz}$,  whereas those of  $B$ are in the order of $\mathrm{MHz}$.   Based on this and  with the aid of Proposition 1, equations \eqref{Eq:Capacity} and \eqref{Eq:Capacity2} can be lower bounded as follows: 
\begin{align}
C_{\rm B} > \frac{1}{2} \sum_{\rm i=1}^{K}\Delta f \log_2\left(1+ \mathcal{B}(\lambda_i) A_0^2\right) - \frac{1}{\xi \ln(2)}
\label{Eq:Capacity1}
\end{align}
and
\begin{align}
C_B > \frac{1}{2} B \log_2\left(1+ \mathcal{B(\lambda)} A_0^2\right) - \frac{B}{\xi \ln(2)}  
\label{Eq:Capacity22}
\end{align}   
respectively. 
As in the previous scenarios, the above inequalities are particularly tight; as a result, they can be also considered as simple and accurate approximations. 

\subsection{Evaluation of the outage probability}
The reliability of the OWCI can be additionally evaluated in terms of  the corresponding outage performance. To this end, we derive  the outage probability of the considered set up in the presence of misalignment fading which is a critical factor in optical wireless communication systems.

\begin{theorem}
The outage probability of the considered system can be expressed as
\begin{align}
P_o(\gamma_{\rm th}) =
\left\{ 
\begin{array}{c c}
\frac{1}{A_0^{\xi}} \left(\frac{2 q R  P_b + 2 q   I_{\rm DC} + N_0}{R^2 \exp\left(-\alpha(\lambda) \delta\right) \tilde{P}_s} \gamma_{\rm th} \right)^{\tfrac{\xi}{2}}, 
& \gamma_{\rm th} \leq \frac{R^2 \exp\left(-\alpha(\lambda) \delta\right) A_0^2 \tilde{P}_s}{2 q R  P_b + 2 q   I_{\rm DC} + N_0},
 \\
 1, & \text{ otherwise. }
\end{array}
\right.
\label{Eq:OP}
\end{align}
\end{theorem}
\begin{proof}
The proof is provided in Appendix~D. 
\end{proof}

From~\eqref{Eq:OP}, it is evident that the maximum spectral efficiency is constrained by the PSD of the noise components, the RX's responsivity, the pathloss, and the intensity of the misalignment effect. 
Furthermore, we observe that as the transmission signal PSD increases, the spectral efficiency constraint relaxes. 

\textcolor{black}{According to~\eqref{Eq:AverageSNR} and~\eqref{Eq:OP}, the outage probability can be alternatively expressed as
\begin{align}
P_o(r_{\rm th}) =
\left\{ 
\begin{array}{c c}
\left(\frac{\xi}{\xi+2} \frac{\gamma_{\rm th}}{\tilde{\gamma}}\right)^{\frac{\xi}{2}}, & \gamma_{\rm th} \leq \frac{\xi+2}{\xi}\tilde{\gamma}, \\
1, & \text{ otherwise }
\end{array}
\right.
\label{Eq:OPvsSNR}
\end{align}
which indicates that for given $\gamma_{\rm th}$ and $\tilde{\gamma}$,  the outage performance of OWCI improves as $A_0$ increases. 
Moreover, it is shown that for a fixed $\gamma_{\rm th}$ and $A_0$,  the outage probability decreases as $\tilde{\gamma}$ increases. 
Finally, it is observed that as $\gamma_{\rm th}$ increases, the outage probability also~increases.}

\subsection{Pointing error displacement tolerance in practical OCWI designs}
It is evident that the derived analytic expressions in the previous sections allow the determination of the pointing error displacement (jitter) and its variance for given quality of service requirements. To this end, it readily follows from \eqref{Eq:OP} that
\begin{equation} \label{New_1a}
\begin{array}{c c}
\sigma_s^2 = \frac{w_{\rm eq}^2}{8 \log_{ \mathcal{H}}(P_o(\gamma_{\rm th}) )} = \frac{w_{\rm eq}^2 \ln (\mathcal{H})}{8 \ln (P_o(\gamma_{\rm th}) )}, 
& \qquad \gamma_{\rm th} \leq \frac{R^2 \exp\left(-\alpha(\lambda) \delta\right) A_0^2 \tilde{P}_s}{2 q R  P_b + 2 q   I_{\rm DC} + N_0}
\end{array}
\end{equation}
and
\begin{equation} \label{New_1b}
\begin{array}{c c}
\sigma_s = \frac{w_{\rm eq} \sqrt{\ln (\mathcal{H})}}{2 \sqrt{2\ln (P_o(\gamma_{\rm th}) )}} = \frac{w_{\rm eq}}{2\sqrt{2}} \sqrt{\frac{\alpha(\lambda) \delta + \ln ( \gamma_{\rm th}) + \ln (2 q R P_b + 2 q I_{\rm DC} + N_0) - 2\ln (A_0) - 2\ln (R) - \ln (\tilde{P}_s)}{\ln (P_o(\gamma_{\rm th}) )}}
\end{array}
\end{equation}
where 
\begin{equation} \label{New_2}
\mathcal{H} = \frac{2 q R P_b + 2 q I_{\rm DC} + N_0}{A_0^2 R^2 \exp\left(-\alpha(\lambda) \delta\right) \tilde{P}_s} \gamma_{\rm th}.
\end{equation}

The above expressions provide useful insights into the practical design of OCWI, since they allow the quantification of the amount of pointing error displacement, and its variance, that can be tolerated with respect to certain outage probability requirements.
More specifically, the system parameters, in the practical designs, will be initially selected in order to achieve a target outage probability that is slightly better than the target one, under the assumption of no pointing error displacement.
Then, using \eqref{New_1b} and setting in it the exact target outage probability  will determine the amount of pointing error displacement that can be tolerated for this particular system and quality of service requirements.
Conversely, for pointing error displacement levels that constrain the target quality of service requirements, these equations can assist in determining the required countermeasures for the incurred detrimental effects.

Finally, the above expressions allow us to express the ergodic capacity representation in \eqref{Eq:Capacity2} in terms of the corresponding outage probability, namely
\begin{align}
C_B= \frac{B}{2} \log_2\left(1+ \mathcal{B(\lambda)} A_0^2\right) - \frac{B A_0^{2} \mathcal{B(\lambda)}}{2 \ln(2)} \Phi\left(-A_0^2 \mathcal{B(\lambda)}, 1, 1+\frac{\ln (P_o(\gamma_{\rm th}) )}{\ln (\mathcal{H})}\right)
\label{Eq:Capacity2b}
\end{align}
which in turn allows the quantification of the required ergodic capacity for certain design parameters and target outage probability requirements. 
In addition, further insights on the role of the involved parameters on the overall system performance can be obtained by tightly lower bounding \eqref{Eq:Capacity2b} with the aid of Proposition 1, namely 
\begin{align}
C_B > \frac{B}{2} \log_2\left(1+ \mathcal{B(\lambda)} A_0^2\right) -  \frac{B\ln (\mathcal{H})}{2 \ln(2)\ln (P_o(\gamma_{\rm th})}.  
\label{Eq:Capacity2b_d}
\end{align}

It is noted that similar expressions can be deduced if  \eqref{Eq:OP} is solved with respect to other involved design parameters. 

\section{Numerical results \& discussion}\label{S:Nr}
In this section, we evaluate the feasibility and effectiveness of OWCI by providing analytical  results and insights  for different  scenarios of interest as well as sets of design parameters.  
The validity of the offered analytical results is extensively verified by respective results from Monte Carlo simulations. 
Unless otherwise   stated,  we henceforth assume that the photodiode (PD) effective area, $A=\pi\beta^2$,  with $\beta$ denoting the radius of the RX's circular aperture, is $1\text{ }\mathrm{mm}^2$. 
Also, we assume that   the divergence angle, $\theta$, is~$20^o$, whilst the skin thickness, $\delta$, is assumed to be equal to~$4\text{ }\mathrm{mm}$ and the noise optical PSD, $N_{\rm 0}$, is set to~$\left(1.3 \text{ }\mathrm{pA/\sqrt{Hz}}\right)^2$~\cite{D:Max3657}. 
The beam waist at distance $\delta$ is determined ~by $w_{\rm \delta}=\delta\tan\left(\theta / 2\right)$, and the pointing error displacement variance, $\sigma_{\rm s}$, is assumed to be~$0.5\text{ }\mathrm{mm}$.
Moreover, according to~\cite{gil2012feasibility} the TOL exhibits remarkably  high immunity to external interference.
In addition, in the following analysis we ignore the noise from external light sources by considering a dark-shielded receiver-to-skin interface.
As a result, the background optical power can be omitted, i.e., $P_{\rm b}=0$.
Furthermore, the PD's dark current, $I_{\rm DC}$, is set to~$0.05\text{ }\mathrm{nA}$ whereas the quantum efficiency of the PD,~$\eta$, is  $0.8$~\cite{B:Transdermal_optical_communications}. 
Finally, the transmitted signal power, $P_{\rm s}$, and PSD, $\tilde{P}_{\rm s}$, are  $0.1\text{ }\mathrm{\mu W}$ and $0.01\text{ }\mathrm{\mu W/MHz}$, respectively, the communication bandwidth, $B$, is $10\text{ }\mathrm{MHz}$, the data rate threshold, $r_{\rm th}$, is set to $=1\text{ }\mathrm{bits/s/Hz}$, the operation wavelength, $\lambda$, is $=1100\text{ }\mathrm{nm}$ and the square ratio between the equivalent beam radius and the pointing error displacement standard deviation (SD) at the RX, $\xi$, is assumed to be $1$.

\textcolor{black}{The remainder of this section is organized as follows: 
The impact of the skin thickness on the performance of the OWCI is described in Section~\ref{Ss:SkinThickness}, while 
 the degradation of the TOL due to the misalignment fading is presented in Section~\ref{Ss:Misalignment}.
The joint effect of the skin thickness and the misalignment  is  provided in Section~\ref{Ss:SkinThicknessAndMisalignment}.
Finally, the impact of various design parameters is analyzed in Section~\ref{Ss:DesignParameters}.}

\begin{figure}[htbp]
\begin{subfigure}[t]{0.48\textwidth}
\centering\includegraphics[width=1\linewidth,trim=0 0 0 0,clip=false]{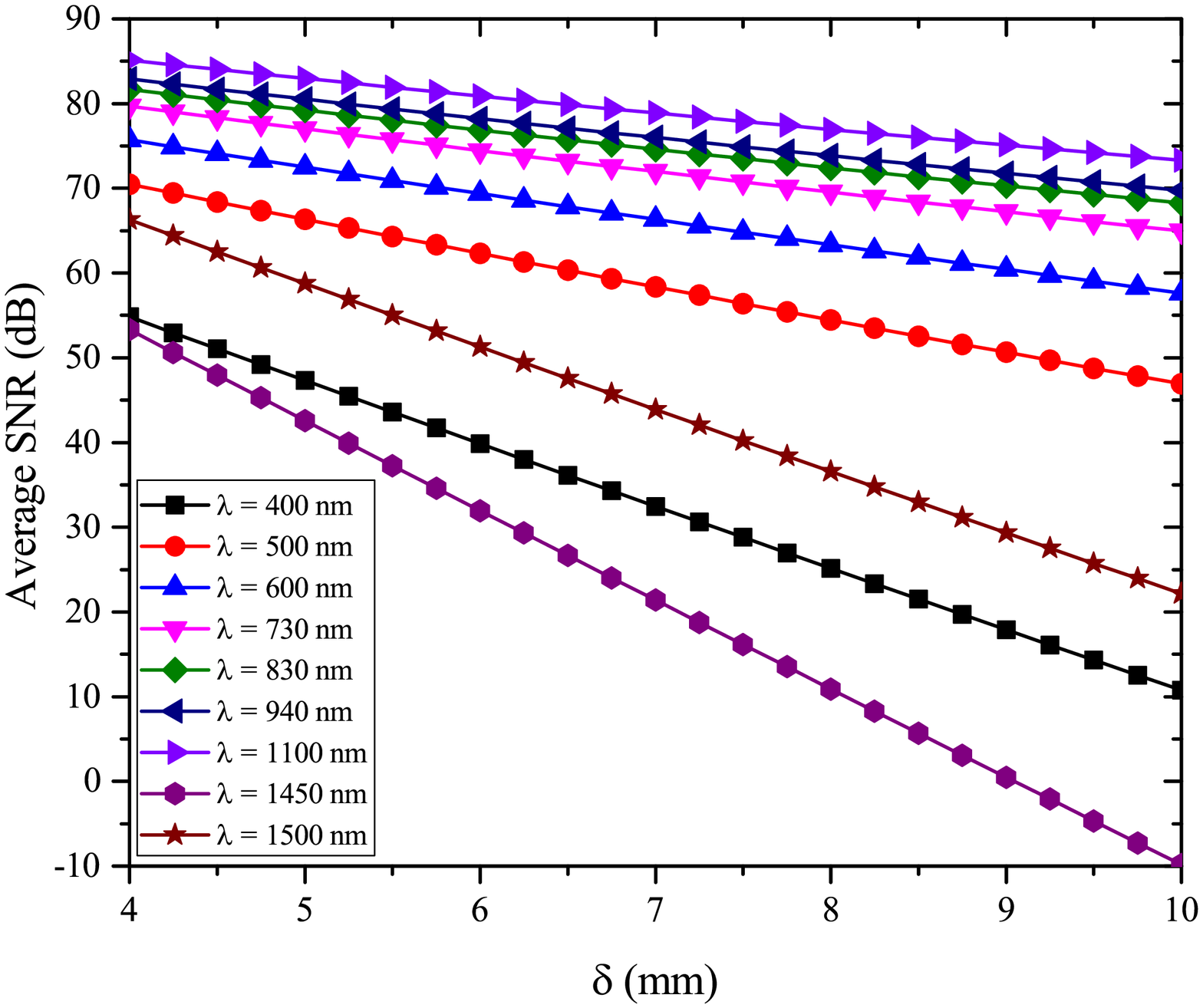}
\caption{Average SNR vs skin thickness for different values of wavelength.}\label{fig:AvSNR_vs_Delta}
\end{subfigure}
\hspace{0.02\textwidth}
\begin{subfigure}[t]{0.48\textwidth}
\centering\includegraphics[width=1\linewidth,trim=0 0 0 0,clip=false]{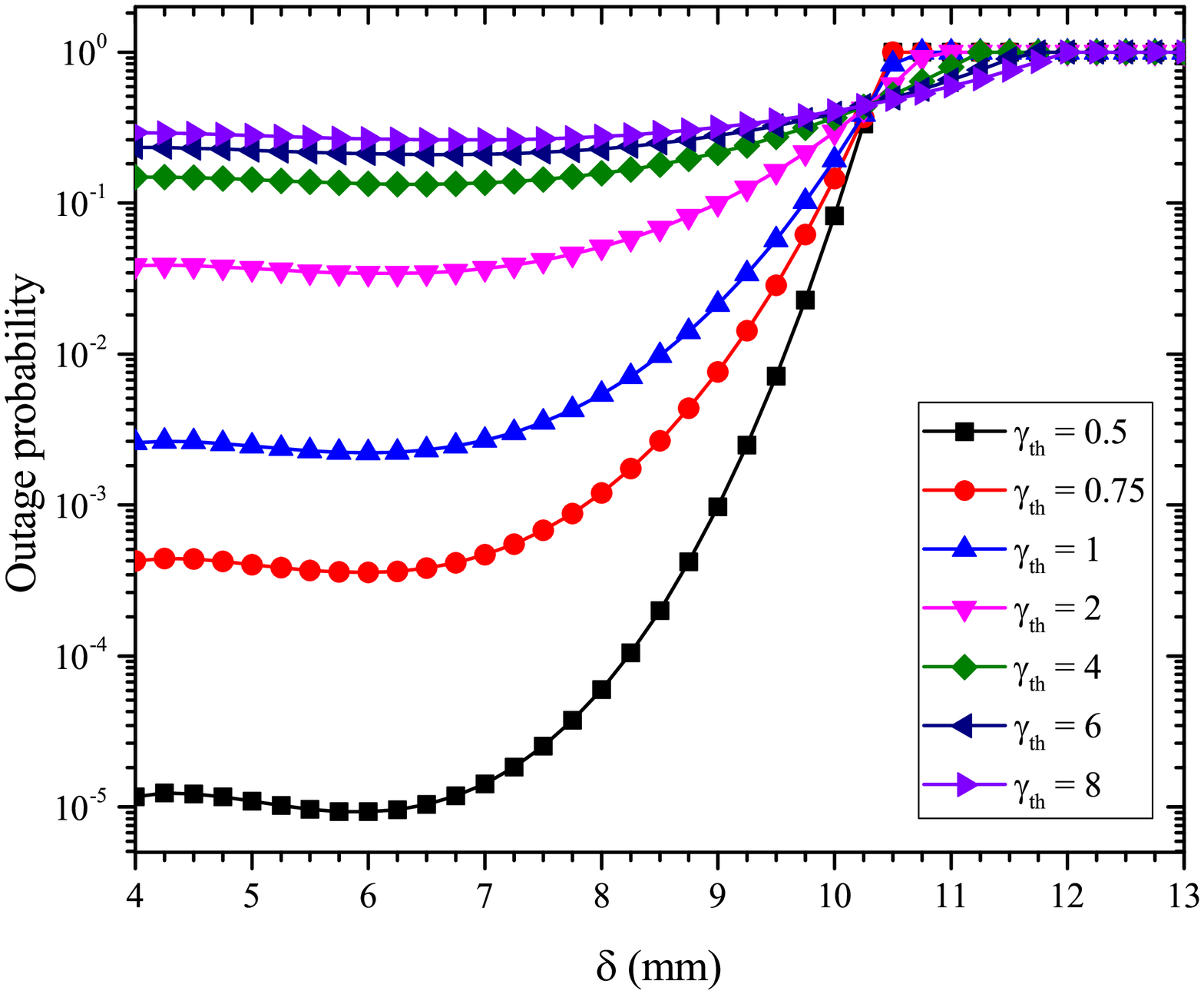}
\caption{Outage probability vs skin thickness for different values of SNR threshold, for $\lambda=1500\text{ }\mathrm{nm}$.}\label{fig:OP_vs_Delta_rth}
\end{subfigure}
\begin{subfigure}[t]{0.49\textwidth}
\centering\includegraphics[width=1\linewidth,trim=0 0 0 0,clip=false]{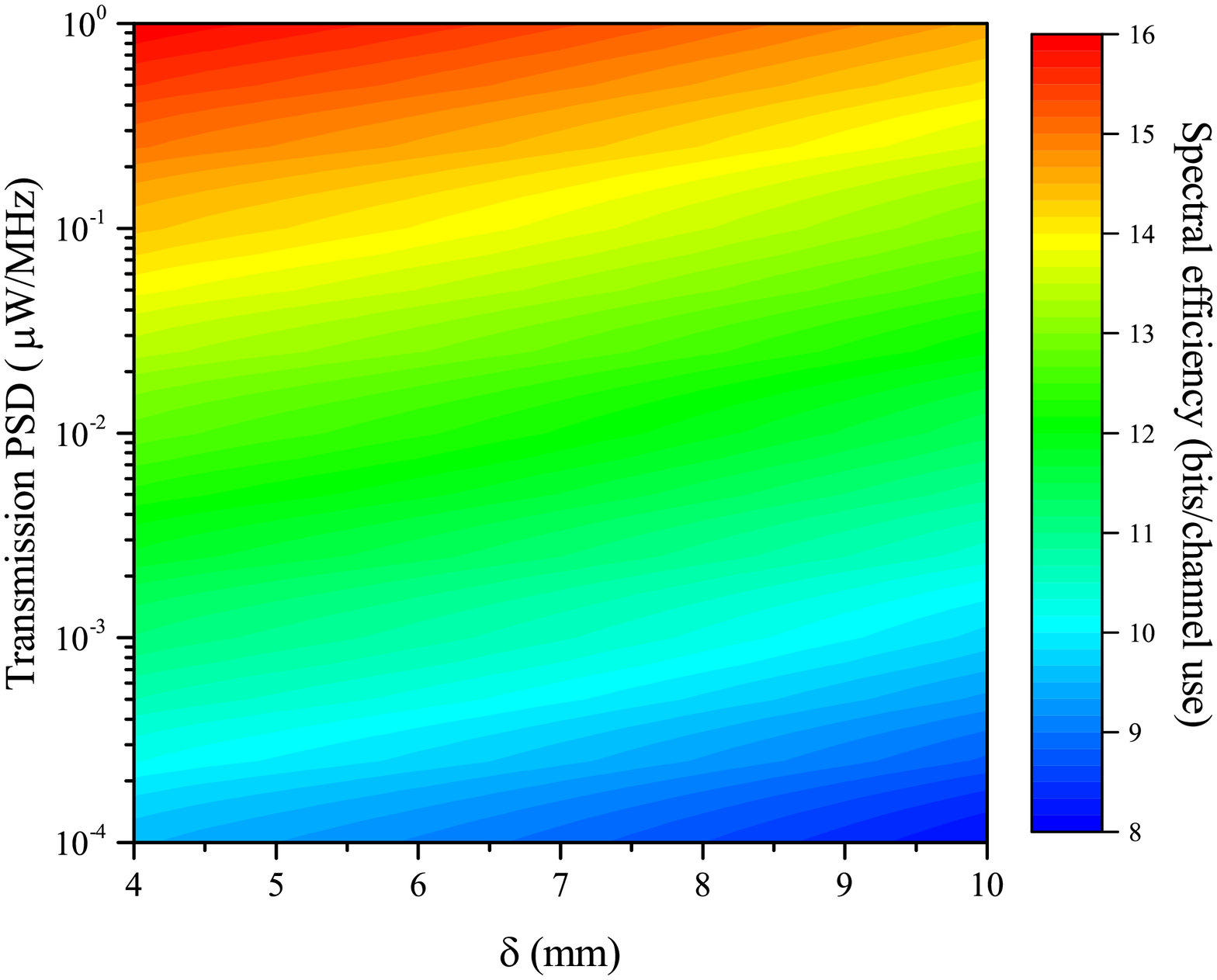}
\caption{Spectral efficiency vs skin thickness and transmission PSD.}\label{fig:Sp_vs_Delta_rth}
\end{subfigure}
\hspace{0.02\textwidth}
\begin{subfigure}[t]{0.48\textwidth}
\centering\includegraphics[width=1\linewidth,trim=0 0 0 0,clip=false]{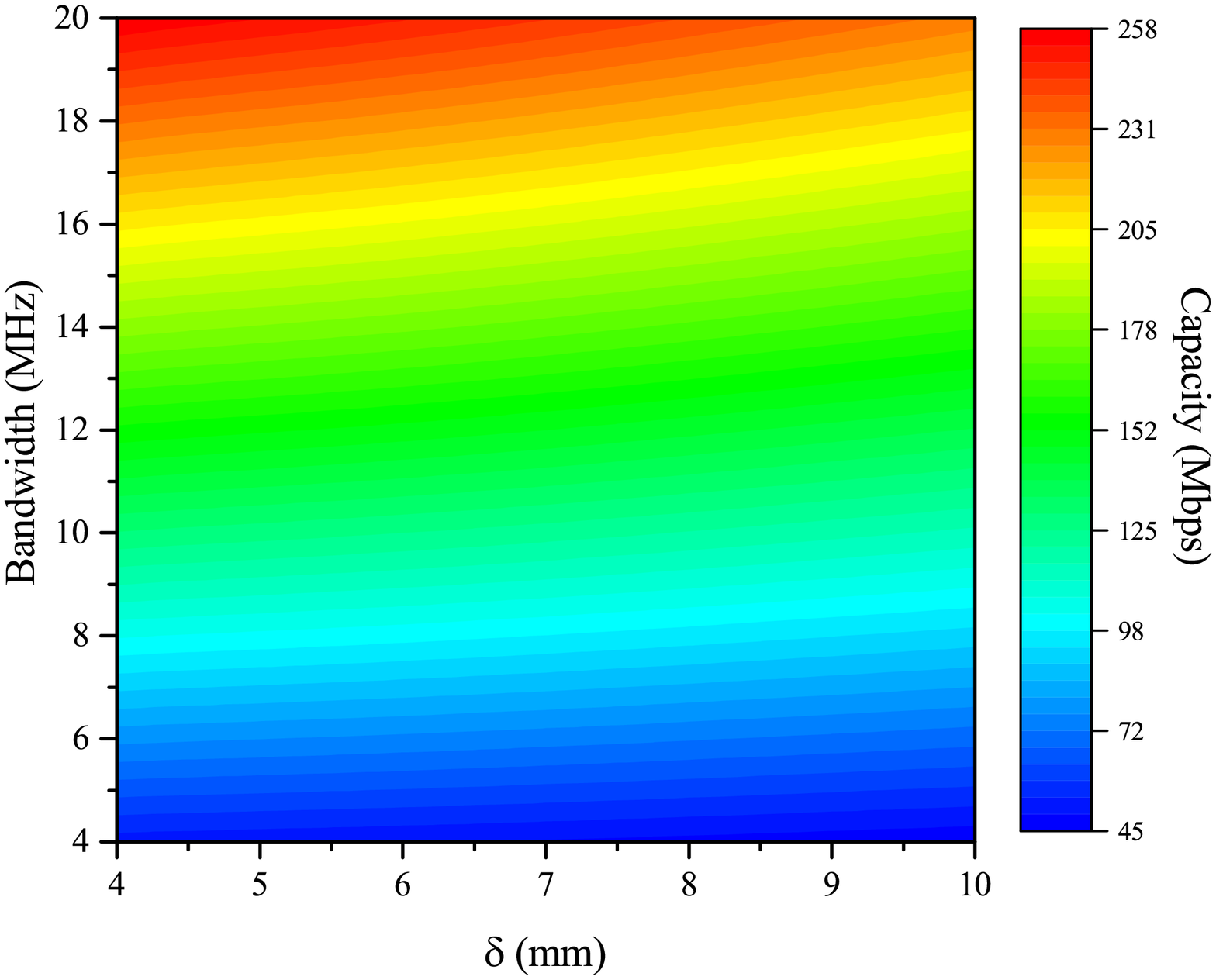}
\caption{Capacity vs skin thickness and bandwidth.}\label{fig:3D_C_Delta_BW}
\end{subfigure}
\caption{Impact of skin thickness on the OWCI's effectiveness.}
\end{figure}

\subsection{Impact of skin thickness}\label{Ss:SkinThickness}
Figure~\ref{fig:AvSNR_vs_Delta} demonstrates the impact of skin thickness on the received signal quality in OWCIs operating at different wavelengths. 
It is observed  that the analytical and simulation results coincide, which verifies   the accuracy of the derived  analytical framework. 
As expected, as the skin thickness increases  for a given wavelength, the average SNR decreases i.e. the received signal quality degrades.
For instance, the average SNR degrades by $66.6\%$ at $\lambda=1500\text{ }\mathrm{nm}$ as $\delta$ changes from $4$ to $10\text{ }\mathrm{mm}$, whereas for $\lambda=1100\text{ }\mathrm{nm}$ and the same $\delta$ variation, the average SNR degrades by $13.9\%$.
This degradation is caused due to the skin pathloss. 
Additionally, we observe that for a given skin thickness, $\delta$, the average SNR depends on the wavelength. 
For example, for $\delta=6$ mm and $\lambda=1450$ nm, the average SNR is $32$ dB, whereas for the same $\delta$ and $\lambda=1500$ nm, the average SNR is $51.3$ dB. 
Thus,  the slight  $50$ nm variation of the wavelength from $1450$ to $1500$ nm  resulted in a $60\%$ signal quality improvement. 
On the contrary,  when $\lambda$ alters from $400$ nm to $500$ nm,  for the same $\delta$, the average SNR increases by  $56.5\%$.
This observation indicates the wavelength dependency of the TOL as well as the importance of appropriately choosing the wavelength during the OWCI system design phase.

Figure~\ref{fig:OP_vs_Delta_rth} illustrates the outage probability as a function of skin thickness for different values of SNR threshold, $\gamma_{\rm th}$, for $\lambda=1500\text{ }\mathrm{nm}$.
We observe that for a fixed skin thickness, the outage probability also increases as $\gamma_{\rm th}$ increases.
For example, for $\delta=6\text{ }\mathrm{mm}$ and $\gamma_{\rm th}=0.5\text{ }\mathrm{bits/s/Hz}$, the outage probability is about $10^{-5}$, whilst for  $\gamma_{\rm th}=4$ and the same $\delta$ it is about $10^{-1}$.
Furthermore, for the same $\gamma_{\rm th}$ and when $\delta$ changes from $4\text{ }\mathrm{mm}$ to $6\text{ }\mathrm{mm}$, it is noticed that the outage probability decreases.
On the contrary,  when $\delta$ exceeds $6\text{ }\mathrm{mm}$ and for  the same $\gamma_{\rm th}$, the outage probability increases.
This is because the misalignment fading drastically affects the transmission at relatively long TX-RX distances. 
In other words, as the skin thickness increases, the intensity of the misalignment fading decreases, whereas the impact of pathloss becomes more severe.

Figure~\ref{fig:Sp_vs_Delta_rth} illustrates  the spectral efficiency as a function of the skin thickness and the PSD of the  transmission signal. 
As expected, increasing the transmission signal PSD can countermeasure the impact of skin thickness. 
For example, for $\delta=8\text{ }\mathrm{mm}$, an increase of the transmission PSD from $0.001\text{ }\mathrm{\mu W/MHz}$ to $0.01\text{ }\mathrm{\mu W/MHz}$  results in a $56.5\%$ increase of the corresponding spectral efficiency.
Moreover, we observe that in the worst case scenario, in which the skin thickness is  $10\text{ }\mathrm{mm}$, even with particularly  low transmission signal PSD of about  $0.0001\text{ }\mathrm{\mu W/MHz}$) the achievable spectral efficiency is in the order of $8\text{ }\mathrm{bits/channel\text{ } use}$.

In Fig.~\ref{fig:3D_C_Delta_BW}, the channel capacity is depicted as a function of the skin thickness and the bandwidth of the system. 
As expected,  an increase of the bandwidth for the same skin thickness results in an increase of the TOL's capacity.
For instance,  as bandwidth increases from $14\text{ }\mathrm{MHz}$ to $18\text{ }\mathrm{MHz}$, for $\delta=6\text{ }\mathrm{mm}$, an improvement of about $28.6\%$ is noticed. 
In addition, it is observed that even in the worst case scenario in which the bandwidth and the skin thickness are considered  $4\text{ }\mathrm{MHz}$ and $10\text{ }\mathrm{mm}$, respectively, the OWCI outperforms the corresponding RF counterparts.

\begin{figure}[htbp]
\begin{subfigure}[t]{0.48\textwidth}
\centering\includegraphics[width=1\linewidth,trim=0 0 0 0,clip=false]{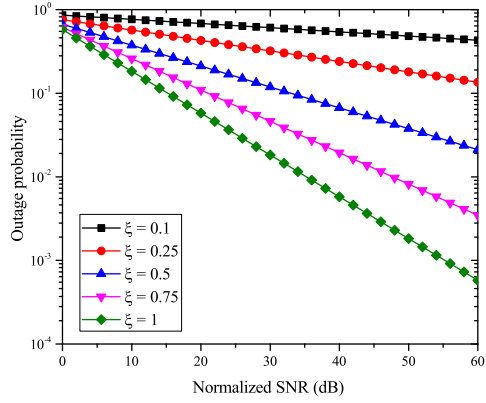}
\caption{Outage probability vs normalised SNR for different values of~$\xi$.}\label{fig:Pout_vs_SNR}
\end{subfigure}
\hspace{0.02\textwidth}
\begin{subfigure}[t]{0.48\textwidth}
\centering\includegraphics[width=1\linewidth,trim=0 0 0 0,clip=false]{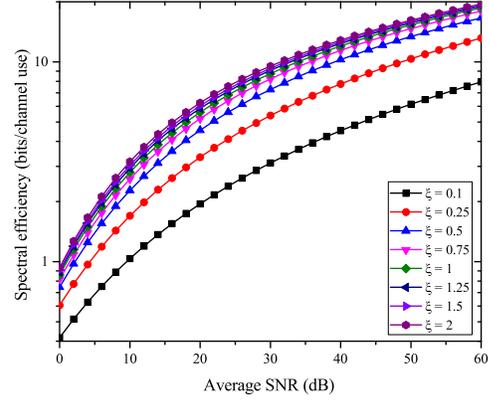}
\caption{Spectral efficiency vs average SNR for different values of~$\xi$.}\label{fig:C_vs_SNR}
\end{subfigure}
\begin{subfigure}[t]{0.51\textwidth}
\centering\includegraphics[width=1\linewidth,trim=0 0 0 0,clip=false]{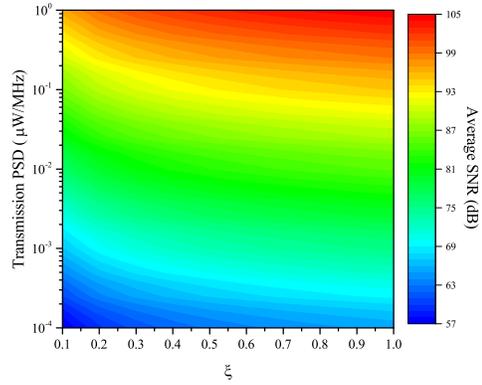}
\caption{Average SNR vs $\xi$ and transmission~PSD.}\label{fig:Av_vs_xi_Ps}
\end{subfigure}
\begin{subfigure}[t]{0.46\textwidth}
\centering\includegraphics[width=1\linewidth,trim=0 0 0 0,clip=false]{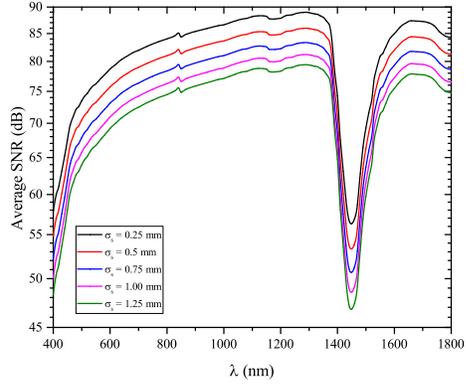}
\caption{Average SNR vs wavelegth  for different values of spatial~jitter.}\label{fig:AvSNR_vs_Lambda_Jitter}
\end{subfigure}
\caption{Impact of misalignment on OWCI's performance.}
\end{figure}

\subsection{Impact of misalignment}\label{Ss:Misalignment}
In Fig.~\ref{fig:Pout_vs_SNR}, the outage probability is illustrated as a function of the normalized SNR for different values of $\xi$. 
In this case  the normalized SNR is defined as the ratio of the average SNR to the SNR threshold, $\tilde{\gamma}/\gamma_{\rm th}$. 
Also,  markers represent the simulation outcomes whereas the  continuous curves represent the corresponding analytical results. 
It is evident that the analytical results are in agreement with the corresponding simulation results, which verifies  the validity of the presented theoretical framework. 
In addition, we observe that for a fixed normalized SNR, the outage probability increases as the misalignment effect increases, i.e., it is inversely proportional to  $\xi$.  
Moreover,  the OWCI performance in terms of  outage probability, for a fixed $\xi$, improves, as expected, as the normalized SNR increases. 
It is also worth noting that as the normalized SNR increases, the misalignment has a more  detrimental effect on the outage performance of the OWCI system. 
For example, for normalized SNR values of  $10$ dB and $50$ dB, the outage probability degrades by about $76\%$ and $99.5\%$, respectively, when $\xi$ is changed from $0.1$ to $1$.  
This indicates the importance of taking  the impact of misalignment into  consideration when determining the required transmit power and spectral~efficiency.    

Figure~\ref{fig:C_vs_SNR} demonstrates the impact of misalignment fading on the spectral efficiency of OWCI.
Again, markers and continuous curves  account for the corresponding simulation and numerical results, respectively. 
The observed agreement between these results  verifies the validity of equation~\eqref{Eq:Capacity2}.
Furthermore, it is noticed that for a given average SNR value, a decrease of $\xi$  constraints the corresponding spectral efficiency. 
For instance, for an average SNR of $40$ dB, a change of $\xi$ from $2$ to $0.1$ results in a $64.8\%$ spectral efficiency degradation, whereas, for the case of $20\text{ }\mathrm{dB}$, the corresponding spectral efficiency decreases by about ~$68.6\%$. 
This indicates that the impact of misalignment fading is more severe in the high SNR regime. 
In addition, we observe that the spectral efficiency increases proportionally to the average SNR  for fixed values of $\xi$.   
For example, for $\xi=0.5$ the spectral efficiency increases by $100\%$ and $60\%$, when the average SNR changes from $10$ dB to $20$ dB and from $20$ dB to $30$ dB, respectively. 
Therefore, it is  evident that as the average SNR increases, the spectral efficiency increase is~constrained.    

Figure~\ref{fig:Av_vs_xi_Ps} depicts the average SNR as a function of $\xi$ and $\tilde{P}_s$. 
As expected, for  fixed values of $\xi$ the PSD  of the  transmission signal is proportional to the average SNR. 
Moreover, for a given $\tilde{P}_s$, as the level of alignment increases, i.e. as $\xi$ increases, the average SNR also increases.
For example, for $\xi=0.5$, the average SNR increases by about $27.4\%$ as the transmission signal PSD increases from $0.001\text{ }\mathrm{\mu W/MHz}$ to $0.1\text{ }\mathrm{\mu W/MHz}$.
Additionally, we observe that for $\tilde{P}_s=0.001\text{ }\mathrm{\mu W/MHz}$, the average SNR increases by a factor of $3.2\%$ as $\xi$ increases from $0.4$ to $0.8$.
This indicates that the misalignment fading affects the average SNR more subtly than the transmission signal PSD.
Finally,  it becomes evident that the impact of misalignment fading can be countermeasured by increasing the PSD of the transmission signal.

Figure~\ref{fig:AvSNR_vs_Lambda_Jitter}, demonstrates the average SNR as a function of the operating wavelength for different values of the jitter SD, $\sigma_s$. 
As expected, for a given wavelength, as the intensity of the misalignment fading increases i.e. as the jitter SD increases, the signal quality degrades and thus the average SNR decreases.
It  is also noticed that the optimal operation wavelength, which maximizes the average SNR, is $1100\text{ }\mathrm{nm}$ whereas the worst performance is achieved at lower wavelengths than $600\text{ }\mathrm{nm}$ and around $1450\text{ }\mathrm{nm}$. 
In addition, a  suboptimal operation wavelength is around $1300\text{ }\mathrm{nm}$, while it is shown that there is an ultra wideband transmission window from $700$ nm to $1300\text{ }\mathrm{nm}$. 
It is worth noting that several commercial optical units (i.e. LEDs and PDs) exist in this particular  window.  

\begin{figure}[htbp]
\begin{subfigure}[t]{0.5\textwidth}
\centering\includegraphics[width=1\linewidth,trim=0 0 0 0,clip=false]{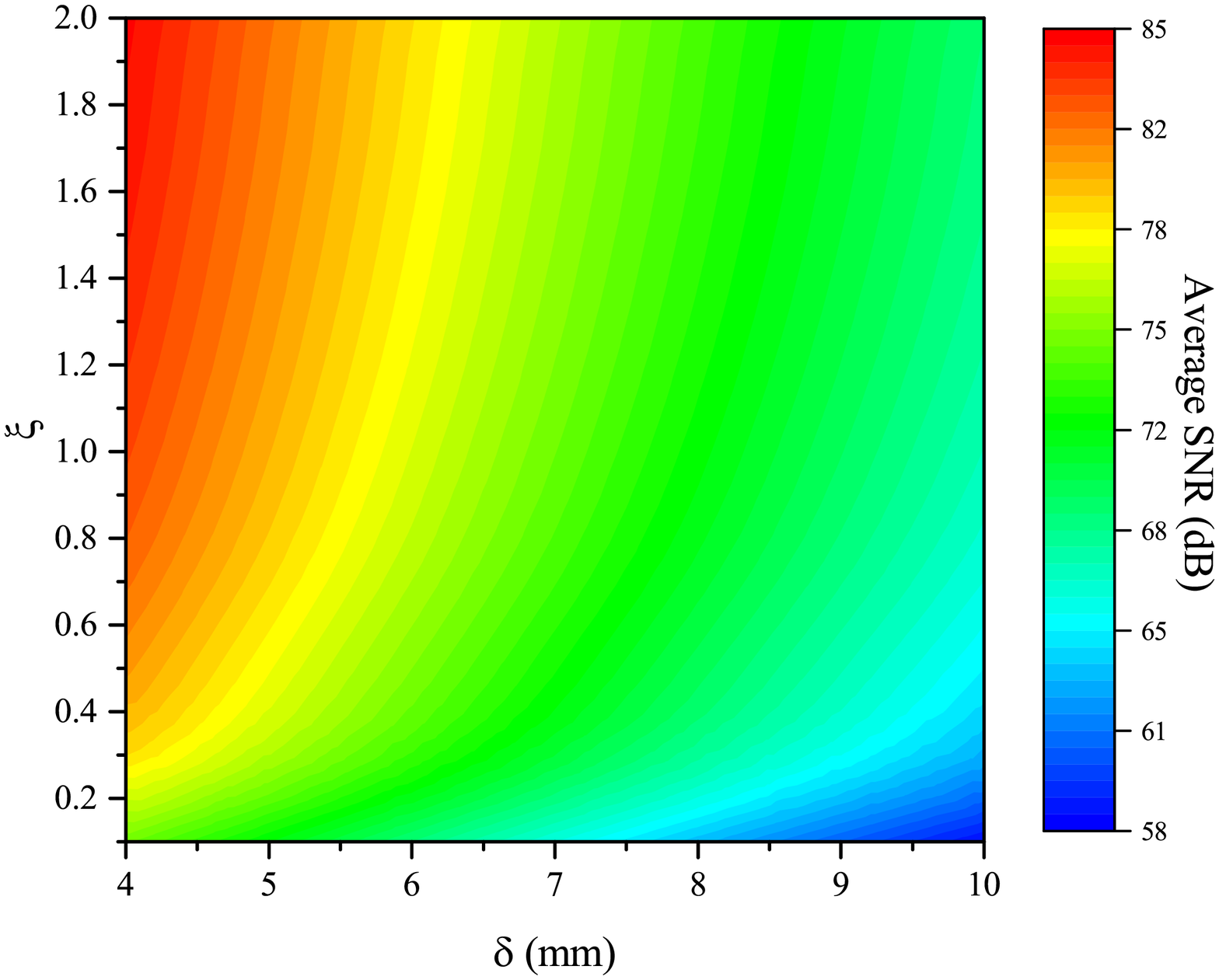}
\caption{Average SNR vs skin thickness and $\xi$, for $\lambda=940\text{ }\mathrm{nm}$.}\label{fig:Delta_Xi_AvSNR}
\end{subfigure}
\hspace{0.02\textwidth}
\begin{subfigure}[t]{0.45\textwidth}
\centering\includegraphics[width=1\linewidth,trim=0 0 0 0,clip=false]{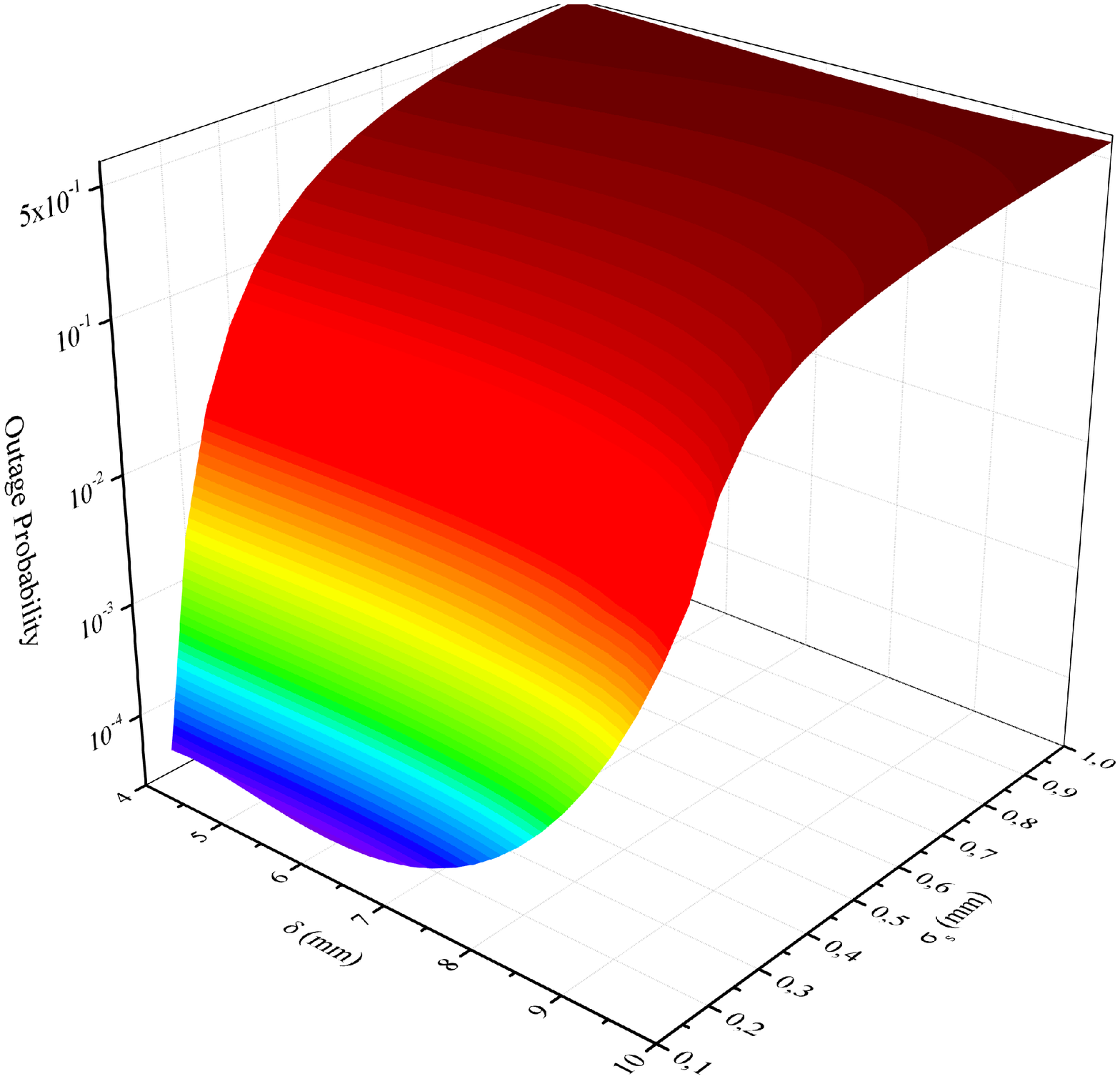}
\caption{Outage probability vs skin thickness and jitter SD.}\label{fig:Delta_ss_AvSNR}
\end{subfigure}
\begin{subfigure}[t]{1\textwidth}
\centering\includegraphics[width=1\linewidth,trim=0 0 0 0,clip=false]{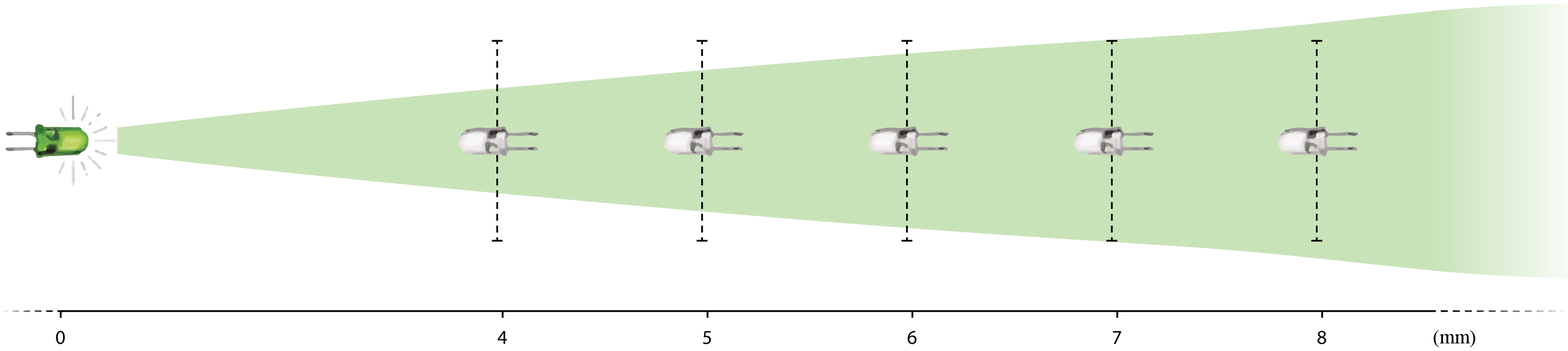}
\caption{Illustration of the misalignment fading as a function of the skin thickness.}\label{fig:jitterIllustration}
\end{subfigure}
\caption{The joint effect of skin thickness and misalignment.\\}
\end{figure}

\subsection{Joint effect of skin thickness and misalignment}\label{Ss:SkinThicknessAndMisalignment}
Figure~\ref{fig:Delta_Xi_AvSNR} illustrates the joint effect of skin thickness and misalignment fading in the received signal quality in terms of the    average SNR and for a    wavelength of  $940$ nm. 
It is evident that for  fixed values of $\xi$, the average SNR decreases  as  the skin thickness increases.
For instance, for $\xi=0.5$, the SNR decreases by $13.8\%$ as the skin thickness increases from $5$ mm to $9$ mm, whereas for $\xi=2$, the SNR decreases by $13.1\%$, for the same skin thickness variation. 
This indicates that the average SNR degradation due to skin thickness increase depends on the intensity of the misalignment fading.     
In addition, we observe that for a fixed skin thickness and as $\xi$ increases i.e. the intensity of misalignment fading decreases, the average SNR also increases. 
For example, for $\delta=7$ mm and when $\xi$ changes  from $0.1$ to $1.5$, the average SNR increases by about $10\text{ }\mathrm{dB}$, from $65.9\text{ }\mathrm{dB}$ to $75.44\text{ }\mathrm{dB}$.
Finally, it is noticed that for practical values of $\delta$ and $\xi$, i.e. $\delta\in[4, 10]\text{ }\mathrm{mm}$ and $\xi\in[0.1, 2]$, the achievable average SNR is particularly high, despite  the low transmission signal PSD.

In Fig.~\ref{fig:Delta_ss_AvSNR}, the outage probability is demonstrated as a function of the skin thickness and the jitter SD. 
It is evident that the impact of misalignment fading  in the low $\delta$ regime is more detrimental than  in the high  regime. 
This is expected since, as illustrated in Fig.~\ref{fig:jitterIllustration} and for a fixed spatial jitter, the skin thickness is proportional to the probability that the RX is within the transmission effective area.  
However, as $\delta$ increases the pathloss also increases, which shows that the impact of misalignment fading is more severe than the effect of pathloss for realistic values of skin thickness. 

\begin{figure}[htbp]
\begin{center}
\begin{subfigure}[t]{0.51\textwidth}
\centering\includegraphics[width=1\linewidth,trim=0 0 0 0,clip=false]{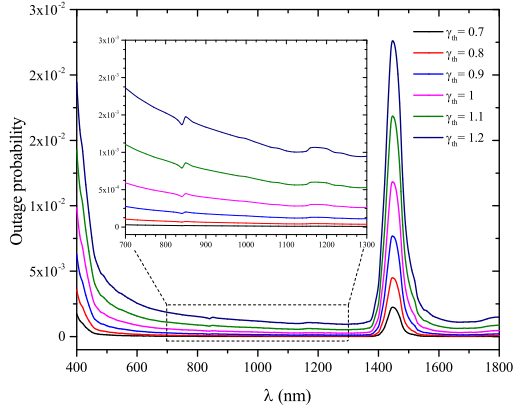}
\caption{Outage probability vs wavelegth for different values of SNR threshold.}\label{fig:OP_lambda}
\end{subfigure}
\end{center}
\begin{subfigure}[t]{0.48\textwidth}
\centering\includegraphics[width=1\linewidth,trim=0 0 0 0,clip=false]{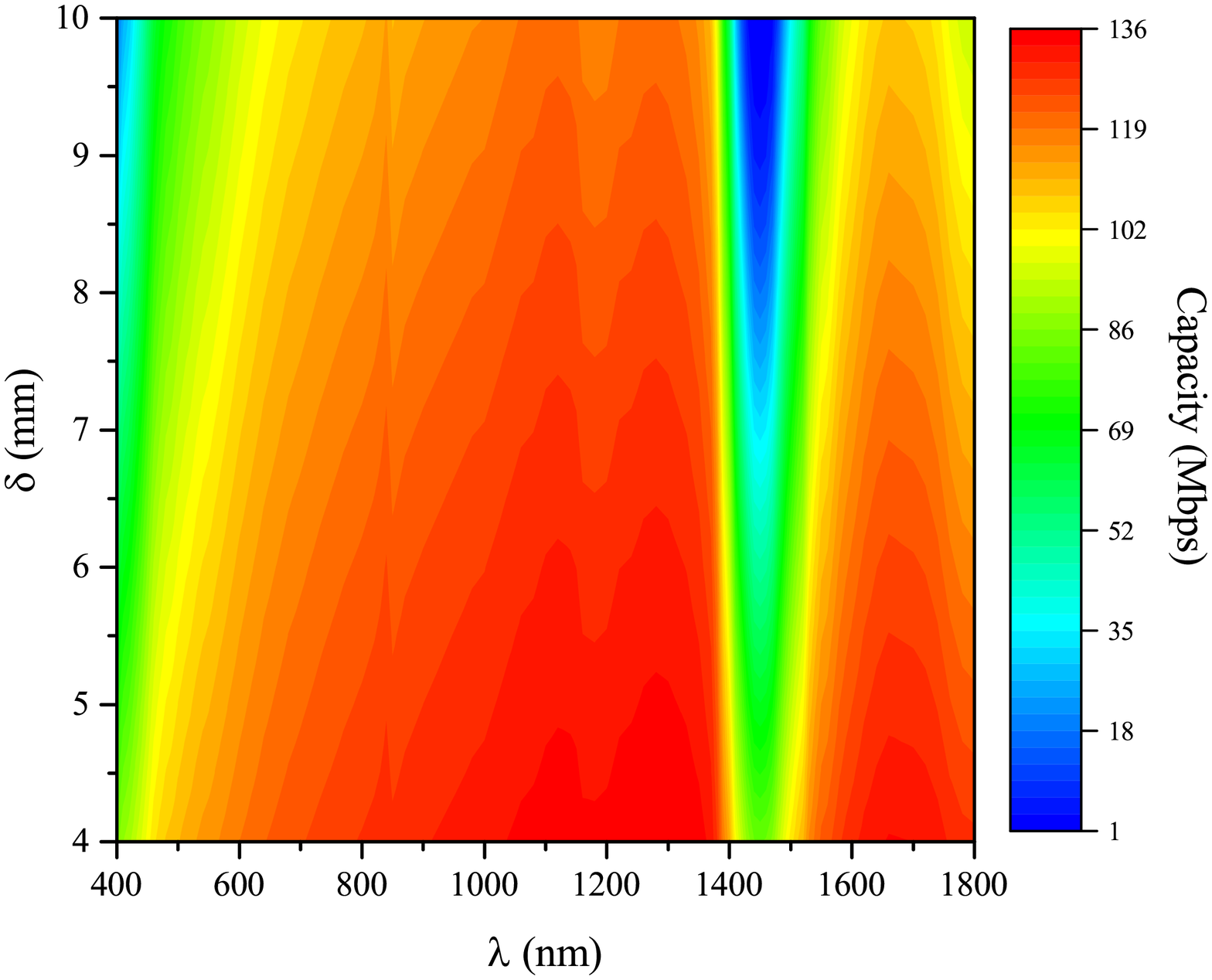}
\caption{Ergodic capacity vs wavelength and skin thickness, assuming heterodyne RX.}\label{fig:C_vs_l_d}
\end{subfigure}
\hspace{0.02\textwidth}
\begin{subfigure}[t]{0.49\textwidth}
\centering\includegraphics[width=1\linewidth,trim=0 0 0 0,clip=false]{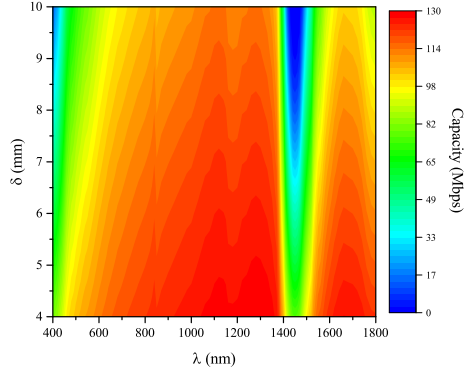}
\caption{Ergodic capacity lower bound vs wavelength and skin thickness, assuming IM/DD RX.}\label{fig:C_vs_l_d_IMDD}
\end{subfigure}
\caption{The impact of channel wavelength selectivity on the performance of OWCIs.}
\end{figure} 

\subsection{Design parameters}\label{Ss:DesignParameters}
In this section, we present the impact of different design selections in the effectiveness of OWCIs in terms of the average SNR, ergodic spectral efficiency and channel capacity. The investigated design parameters   are related to the transmission LED and PD physical characteristics. 

\subsubsection{Selecting the appropriate wavelength}\label{Sss:WavelengthSelectivity}
Figure~\ref{fig:OP_lambda} depicts the outage probability as a function of the wavelength for different values of  the SNR threshold, assuming $\xi=1$ and a transmission signal PSD of~$0.1 \text{ }\mathrm{\mu W/MHz}$. 
Moreover,  the outage probability is, as expected, proportional to $\gamma_{\rm th}$ for fixed wavelength values.  
For example, for the case of $\lambda=600\text{ }\mathrm{nm}$ the outage probability increases approximately by $150\%$ as $\gamma_{\rm th}$ is changed  from $0.8$ to $0.9$. 
Likewise, it is evident that a transmission window exists for wavelengths between $700$ nm and $1300\text{ }\mathrm{nm}$.
It is worth noting that  several commercial LEDs and PDs exist in this wavelength range (see for example\cite{D:TSHG5510, D:IR333-A, D:LED1050E, D:LED1070E, D:LED1200E}).
Finally, we observe that the wavelength windows from $400$ nm to $600\text{ }\mathrm{nm}$ and around $1450\text{ }\mathrm{nm}$ are not optimal for OWCIs, since  the optimal transmission wavelength is $1100\text{ }\mathrm{nm}$.

In Fig.~\ref{fig:C_vs_l_d}, the ergodic capacity of the TOL is demonstrated  as a function of the wavelength and the skin thickness, under the assumption of heterodyne RX. 
Although the transmit power and the bandwidth are relatively low, OWCI can achieve a capacity in the range between $1\text{ }\mathrm{Mbps}$ and $136\text{ }\mathrm{Mbps}$. 
In addition,  we observe that  for a fixed wavelength the skin thickness increases as the capacity decreases. 
For instance, for the case of $\lambda=500\text{ }\mathrm{nm}$ the capacity degrades by $10.3\%$ and $12.3\%$ as the skin thickness increases from $5$ mm to $7\text{ }\mathrm{mm}$ and from $7$ mm to $9\text{ }\mathrm{mm}$, respectively. 
Likewise,   the same skin thickness alterations for  $\lambda=1100\text{ }\mathrm{nm}$  result to a respective capacity degradation of  $3.3\%$ and $4\%$. 
Furthermore,  for the same skin thickness alterations and $\lambda=1450\text{ }\mathrm{nm}$, the capacity degrades by $49.5\%$ and $85\%$, respectively.
This indicates that the level of the capacity degradation due to the skin thickness  depends largely on the corresponding  wavelength value. 
Moreover, it is evident that the optimal performance in terms of spectral efficiency is  achieved at a wavelength  around $1100\text{ }\mathrm{nm}$, independently  of the skin thickness. 
In this particular wavelength, the minimum and maximum achievable ergodic capacities are  $120.87\text{ }\mathrm{Mbps}$ and $135\text{ }\mathrm{Mbps}$, respectively. 
On the contrary, for $\delta\leq 5.25\text{ }\mathrm{nm}$ the worst possible capacity is achieved at $\lambda=400\text{ }\mathrm{nm}$, whereas for $\delta>5.25\text{ }\mathrm{nm}$  it is achieved at $\lambda=1450\text{ }\mathrm{nm}$.
Furthermore, we observe that the optimal transmission window ranges between $\lambda=900$ nm and  $\lambda=1300\text{ }\mathrm{nm}$.

Similarly,  the ergodic capacity lower bound is illustrated   in Fig.~\ref{fig:C_vs_l_d_IMDD} as a function of the wavelength and the skin thickness, assuming IM/DD RX. 
By comparing the results in Fig.~\ref{fig:C_vs_l_d} and~\ref{fig:C_vs_l_d_IMDD}, we observe that in both cases the optimal, the sub-optimal and the worst transmission wavelengths coincide. 
Also,  it is evident that regardless  of the type of the RX, the optimal transmission window is in the range between $\lambda=900$ nm and  $\lambda=1300\text{ }\mathrm{nm}$. 
In other words, the comparison of these figures reveals that the selection of the transmission wavelength should be independent of the RX type. 

\begin{figure}[htbp]
\begin{subfigure}[t]{0.47\textwidth}
\centering\includegraphics[width=1\linewidth,trim=0 0 0 0,clip=false]{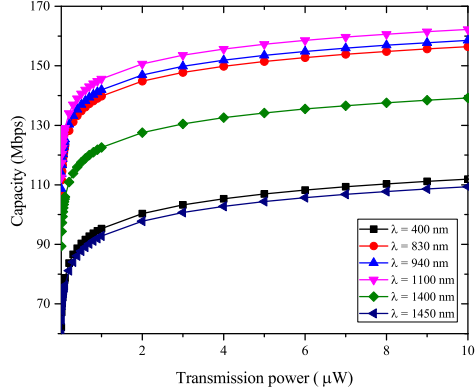}
\caption{Ergodic capacity vs transmit power for different wavelengths.}\label{fig:C_vs_Power}
\end{subfigure}
\hspace{0.01\textwidth}
\begin{subfigure}[t]{0.51\textwidth}
\centering\includegraphics[width=1\linewidth,trim=0 0 0 0,clip=false]{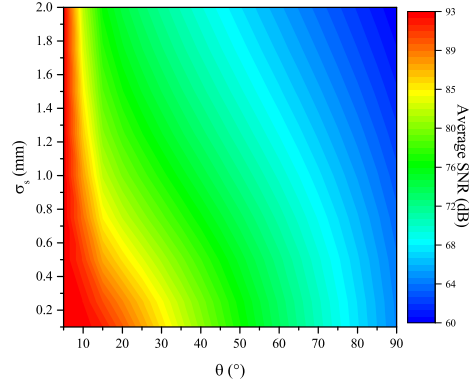}
\caption{Average SNR vs divergence angle and jitter SD.}\label{fig:3D_AvSNR_Theta_Jitter}
\end{subfigure}
\begin{subfigure}[t]{0.48\textwidth}
\centering\includegraphics[width=1\linewidth,trim=0 0 0 0,clip=false]{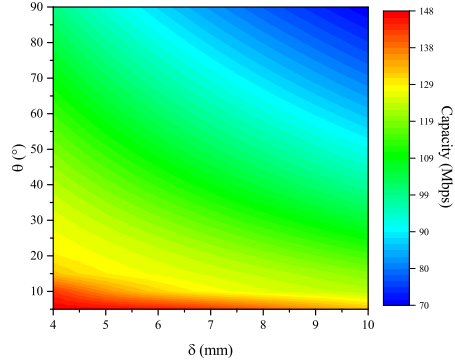}
\caption{Ergodic capacity vs skin thickness and divergence angle.}\label{fig:C_vs_delta_theta}
\end{subfigure}
\hspace{0.02\textwidth}
\begin{subfigure}[t]{0.48\textwidth}
\centering\includegraphics[width=1\linewidth,trim=0 0 0 0,clip=false]{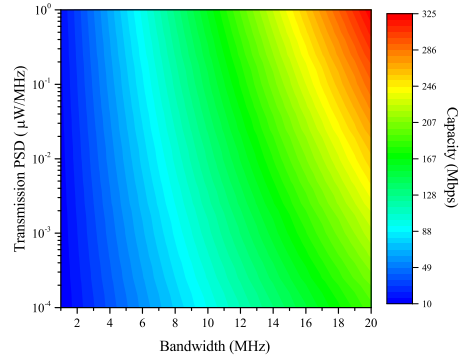}
\caption{Ergodic capacity lower bound vs  bandwidth  and transmit power.}\label{fig:C_vs_BW_Ps}
\end{subfigure}
\caption{The impact of TX design parameters on the performance of OWCI.}
\end{figure}

\subsubsection{TX parameters}\label{Sss:TxParameters}
In Fig.~\ref{fig:C_vs_Power}, the capacity of the TOL is depicted as a function of the transmit power at different wavelengths. 
The considered wavelengths correspond to the best case scenario, which occurs at the optimal wavelength of the OWCI ($1100\text{ }\mathrm{nm}$), the worst case scenario based on the skin attenuation coefficient ($1450\text{ }\mathrm{nm}$), the worst case scenario based on the responsivity of the photodetector ($400\text{ }\mathrm{nm}$), two cases of wavelengths commonly used in commercial LEDs ($830\text{ }\mathrm{nm}$ and $940\text{ }\mathrm{nm}$), and one scenario for an intermediate value ($1400\text{ }\mathrm{nm}$).
To this end, we observe that the maximum capacity is achieved at $1100\text{ }\mathrm{nm}$, while at $830$ nm, where the TSHG5510 operates~\cite{D:TSHG5510}, and $940\text{ }\mathrm{nm}$, where the IR333~\cite{D:IR333-A} operates, the received capacity approaches the optimum.
This indicates that commercial optics can be effectively used in the development of OWCIs. 
Interestingly, the worst case scenario at $1450\text{ }\mathrm{nm}$ provides slightly higher capacity than the one at $400\text{ }\mathrm{nm}$.
As a result, we observe that both the  PD's responsivity and the skin transmitivity drastically impact the quality of the OWCI link.
Also, even though the transmit signal power of the OWCI is considerably lower compared to the RF counterpart (in the order of $40\text{ }\mathrm{mW}$~\cite{A:Cochlear_Implants_System_design_integration_and_evaluation,4431855}), the achievable data rate of the OWCI is surprisingly high, i.e. in the order of $150\text{ }\mathrm{Mbps}$. 
This indicates that the OWCI can achieve a substantially higher power and data rate efficiency compared to the corresponding conventional RF system. 

In Fig.~\ref{fig:3D_AvSNR_Theta_Jitter}, the average SNR is illustrated as a function of the TX's divergence angle and the jitter SD.
We observe that the value of average SNR decreases at an increase of  the divergence angle and the jitter SD.
For instance, for $\theta=20^{o}$, as jitter SD increases from $0.1\text{ }\mathrm{mm}$ to $1\text{ }\mathrm{mm}$,  the average SNR increases by $10.22\%$.
In addition, it is noticed that for a fixed $\sigma_s=0.5\text{ }\mathrm{mm}$ the average SNR increases by about $9.8\%$ as the divergence angle varies from $10^{o}$ to $30^{o}$.
This indicates that by employing LEDs with low divergence angle, we can countermeasure the impact of misalignment fading.

In Fig.~\ref{fig:C_vs_delta_theta}, the ergodic capacity is illustrated as a function of the skin thickness and the divergence angle, assuming heterodyne RX. 
It is shown that for a fixed $\theta$, the capacity decreases as  $\delta$ increases. 
For example, for $\theta=15^o$ and $\delta=5\text{ }\mathrm{mm}$, a $128.77\text{ }\mathrm{Mbps}$ ergodic capacity is achieved, whereas for the same $\theta$ and $\delta=10\text{ }\mathrm{mm}$, the ergodic capacity is $119.43\text{ }\mathrm{Mbps}$. Therefore, a $7.25\%$ ergodic capacity degradation occurs  when the skin thickness doubles.   
On the contrary, for a given $\delta$ the capacity is inversely proportional to $\theta$; hence, we observe that by decreasing the divergence angle, we can countermeasure the incurred pathloss effect. For instance, for the case of  $\delta=8\text{ }\mathrm{mm}$, a $5\%$ capacity increase is achieved when the divergence angle changes from $20^o$ to $10^o$. 

Figure~\ref{fig:C_vs_BW_Ps}, depicts the capacity of the TOL as a function of the transmission bandwidth and PSD.
As expected, for a fixed transmission PSD the capacity of the system improves as the available bandwidth increases.
For example, for $\tilde{P_s}=0.001\text{ }\mathrm{\mu W/MHz}$,  the capacity improves by approximately $200\%$ as the bandwidth increases from $6\text{ }\mathrm{MHz}$ to $18\text{ }\mathrm{MHz}$.
In addition, we observe that for a certain bandwidth, an increase of the transmission PSD improves the system capacity by a lesser factor.
For instance, when the bandwidth is $10\text{ }\mathrm{MHz}$ and transmission PSD increases from $0.001\text{ }\mathrm{\mu W/MHz}$ to $0.1\text{ }\mathrm{\mu W/MHz}$, the TOL's capacity also increases by a factor of $29.6\%$.
This indicates that the bandwidth affects the quality of the TOL  more severely than the transmission PSD.

\begin{figure}[htbp]
\begin{subfigure}[t]{0.48\textwidth}
\centering\includegraphics[width=1\linewidth,trim=0 0 0 0,clip=false]{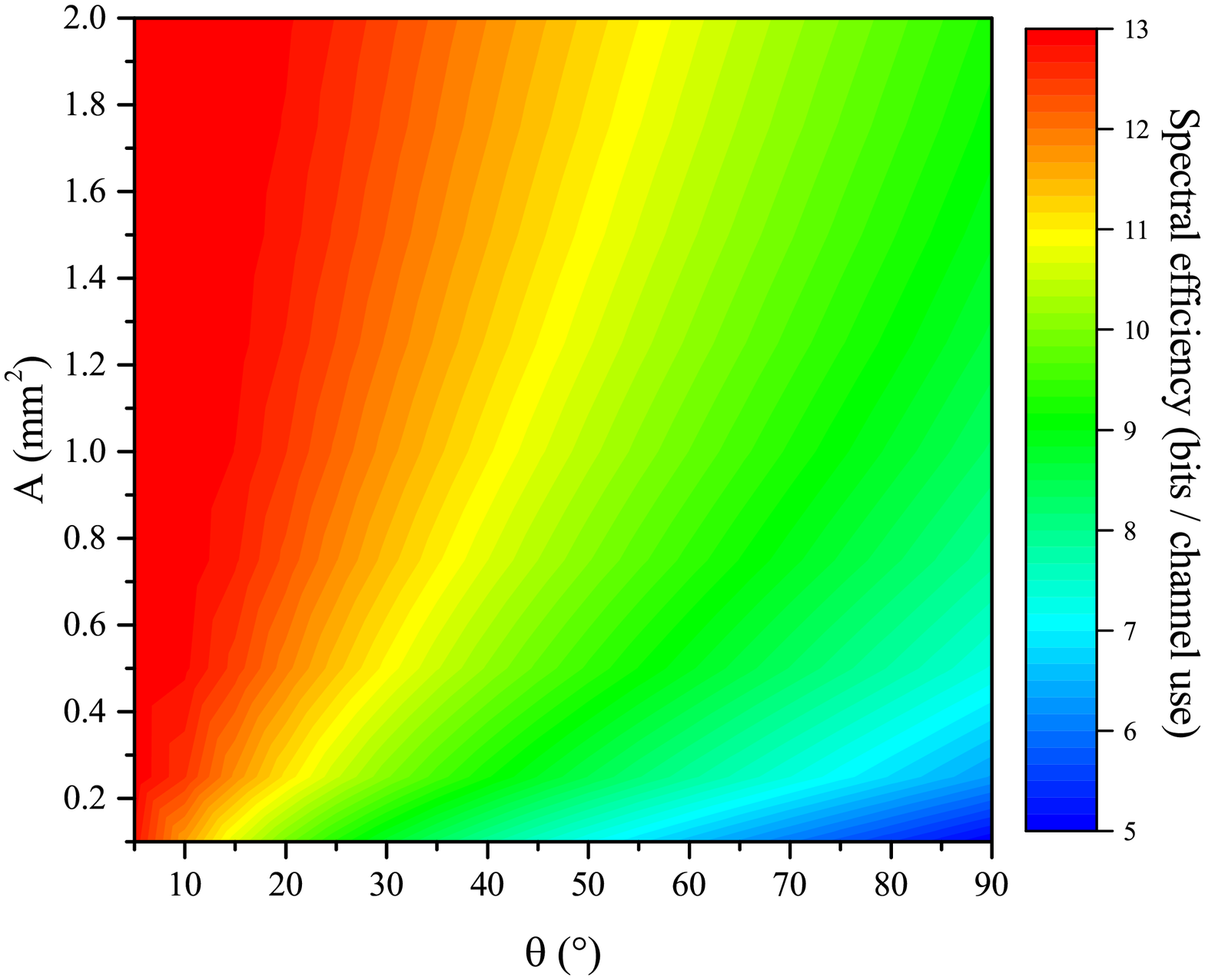}
\caption{Spectral efficiency vs PD effective area and divergence angle.}\label{fig:3D_Theta_Aeff_C}
\end{subfigure}
\hspace{0.02\textwidth}
\begin{subfigure}[t]{0.48\textwidth}
\centering\includegraphics[width=1\linewidth,trim=0 0 0 0,clip=false]{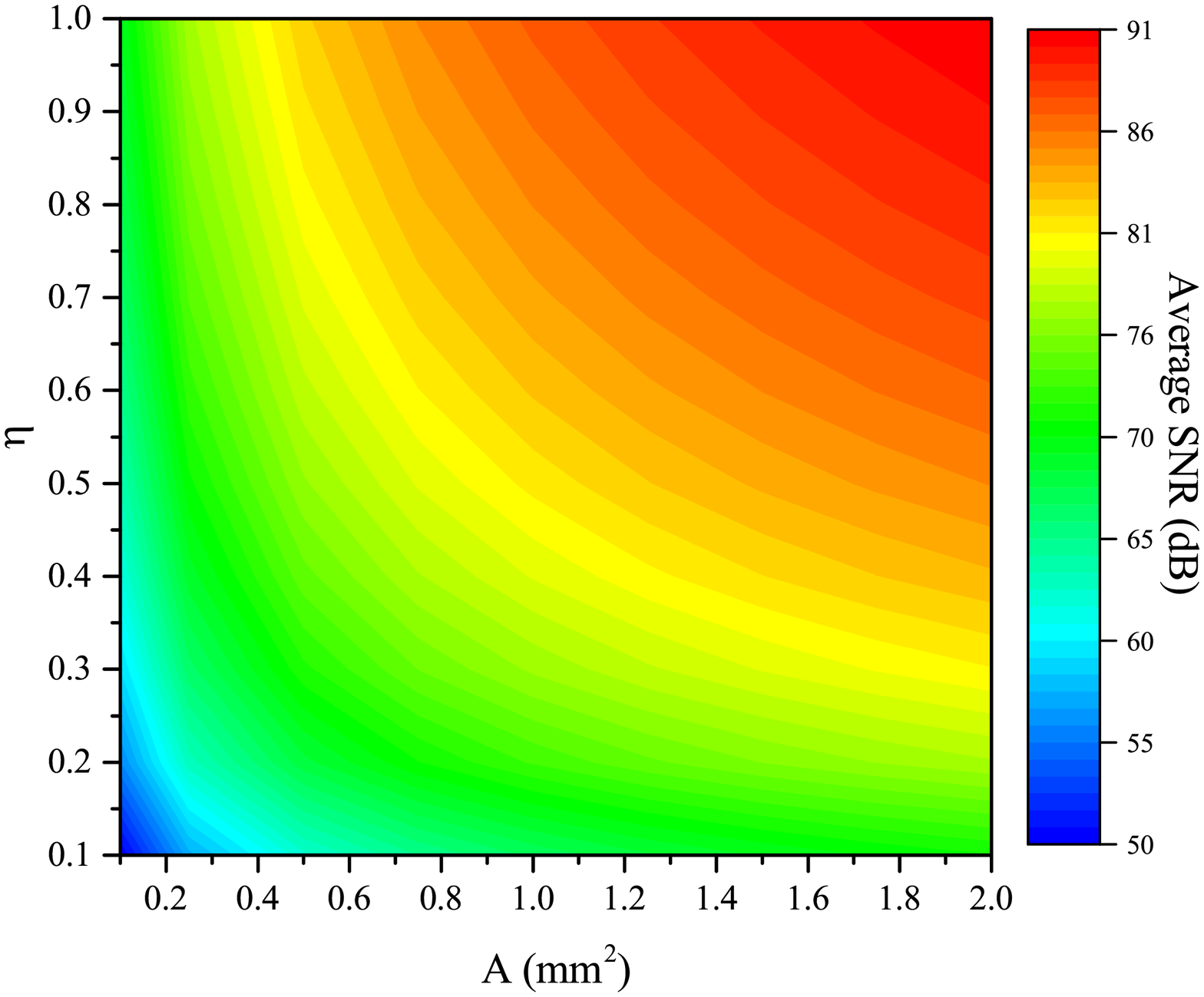}
\caption{Average SNR vs  PD effective area and photodiode's quantum efficiency.}\label{fig:AvSNR_vs_A_eta_3D}
\end{subfigure}
\begin{subfigure}[t]{0.48\textwidth}
\centering\includegraphics[width=1\linewidth,trim=0 0 0 0,clip=false]{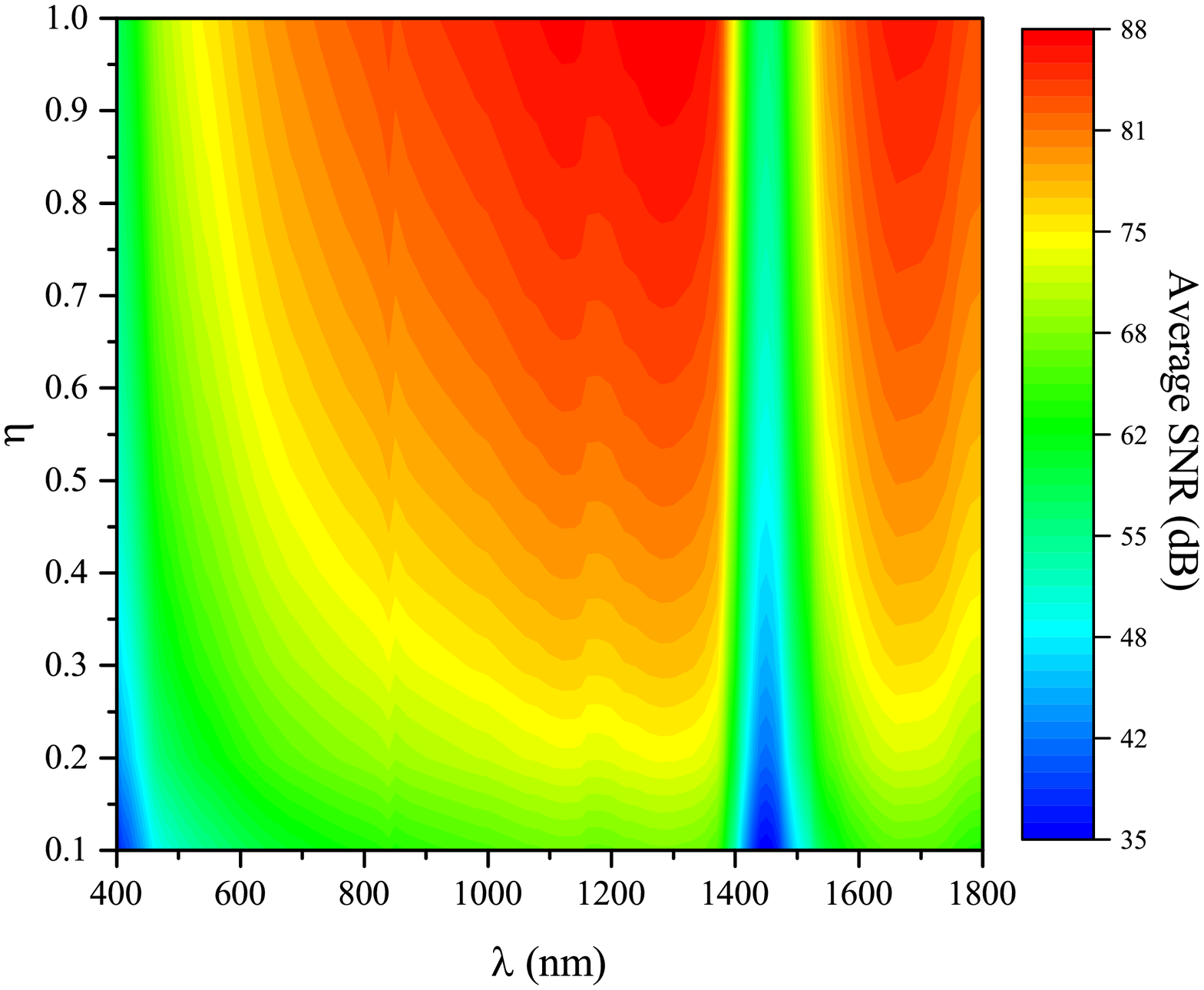}
\caption{Average SNR vs  wavelength and PD's quantum efficiency.}\label{fig:AvSNR_vs_lambda_eta_3D}
\end{subfigure}
\hspace{0.02\textwidth}
\begin{subfigure}[t]{0.48\textwidth}
\centering\includegraphics[width=1\linewidth,trim=0 0 0 0,clip=false]{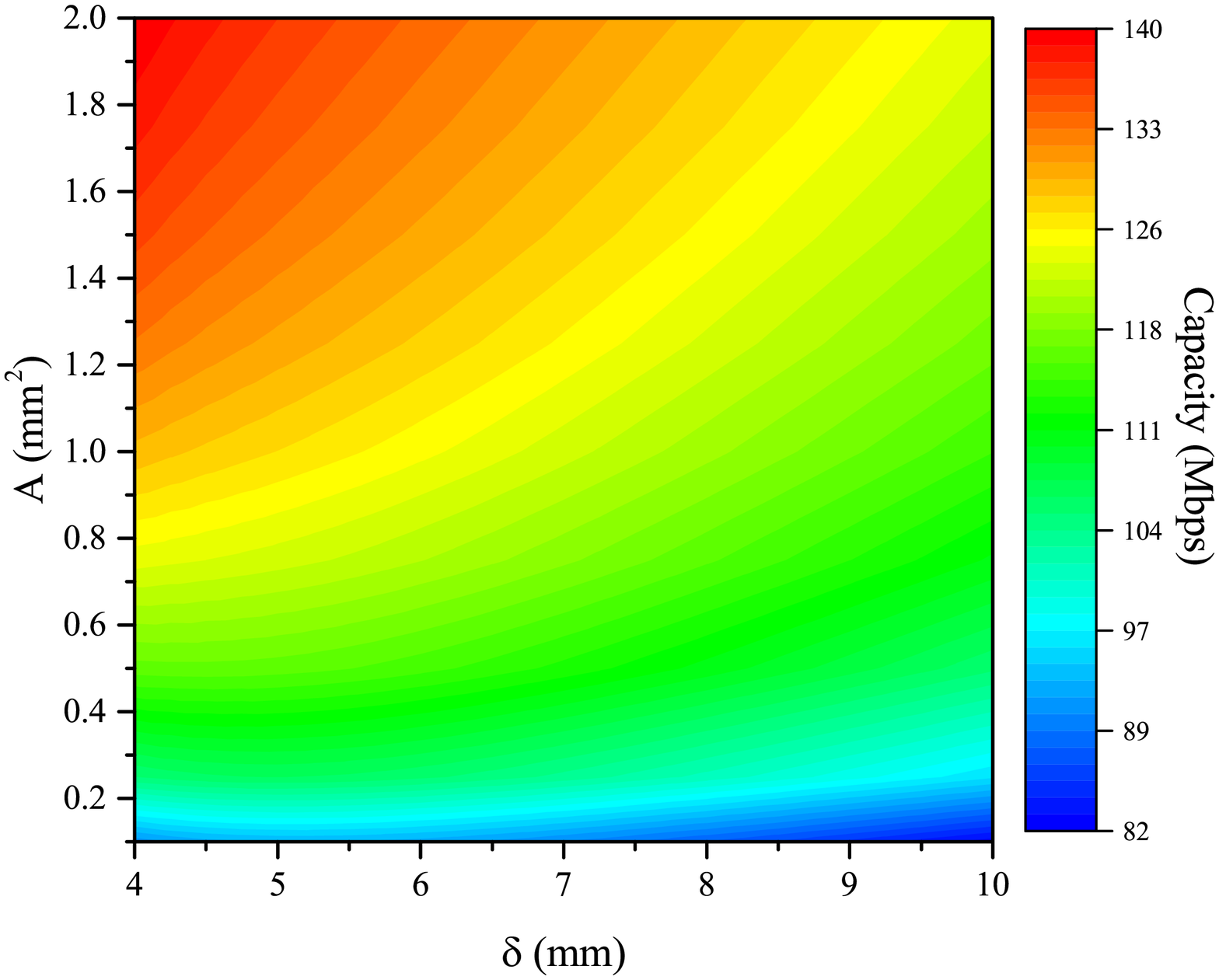}
\caption{Ergodic capacity vs  skin thickness and PD effective area.}\label{fig:C_vs_delta_Aeff}
\end{subfigure}
\begin{subfigure}[t]{0.49\textwidth}
\centering\includegraphics[width=1\linewidth,trim=0 0 0 0,clip=false]{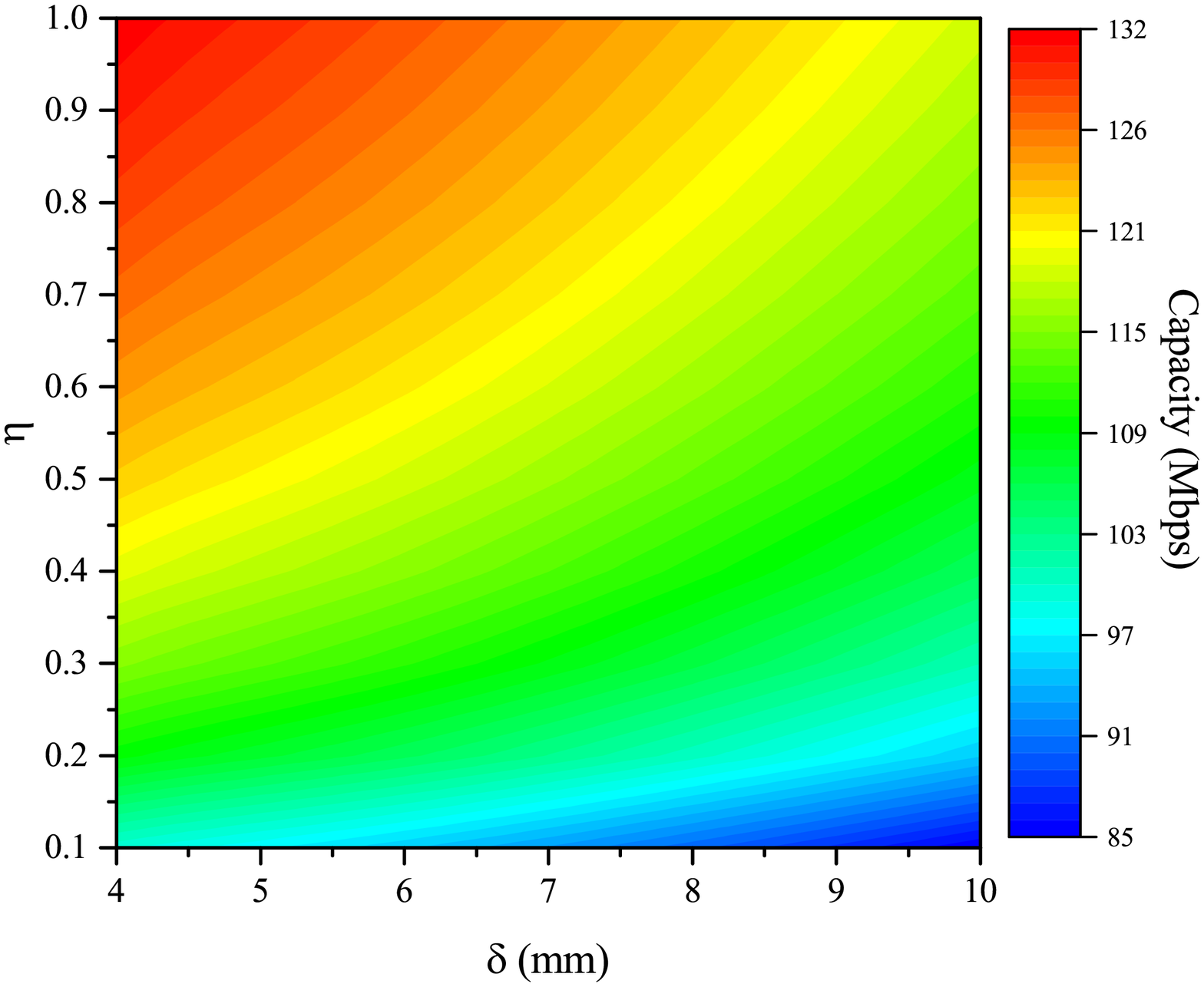}
\caption{Ergodic capacity vs skin thickness and quantum efficiency.}\label{fig:C_vs_eta_delta}
\end{subfigure}
\hspace{0.02\textwidth}
\begin{subfigure}[t]{0.44\textwidth}
\centering\includegraphics[width=1\linewidth,trim=0 0 0 0,clip=false]{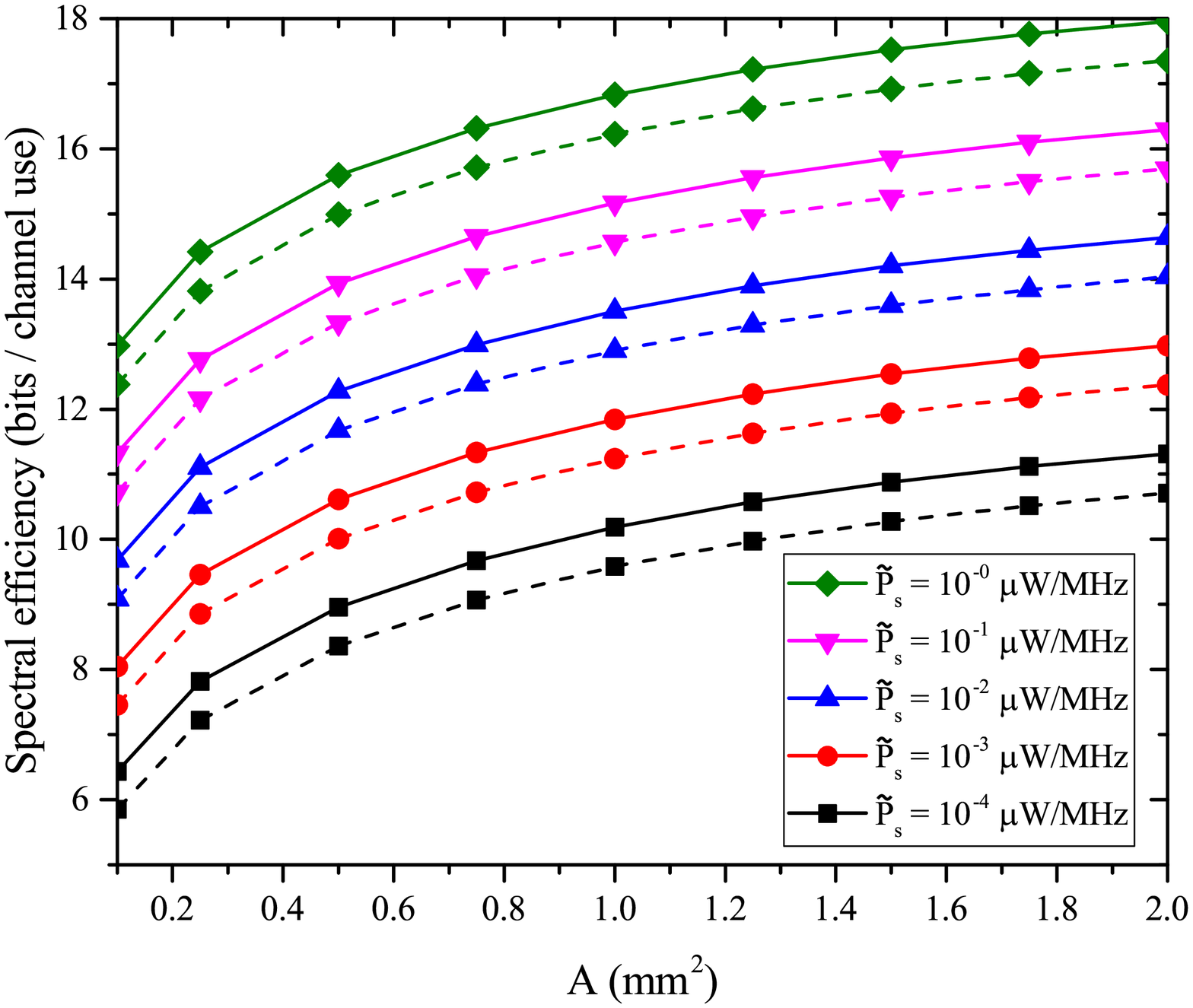}
\caption{Spectral efficiency vs photodiode effective area for different values of transmissioin PSD, heterodyne (continious lines) and IM/DD (dashed lines) RX.}\label{fig:C_vs_Ps_Aeff}
\end{subfigure}
\caption{The impact of RX design parameters on the performance of OWCI.}
\end{figure}

\subsubsection{RX parameters}\label{Sss:RxParameters}
This section discusses the impact of the physical characteristics of the RX optical unit. 
To this end,  Fig.~\ref{fig:3D_Theta_Aeff_C} exhibits the proportionality of the spectral efficiency with respect to the effective area,  $A$,  of the PD for a given divergence angle, $\theta$. 
For instance,  a $24.9\%$ spectral efficiency increase can be achieved at  $\theta=20^{\circ}$ by increasing the PD's effective area from $0.1\text{ }\mathrm{mm}^2$  to $1\text{ }\mathrm{mm}^2$.
On the contrary,  the spectral efficiency is inversely proportional to  $\theta$ for a fixed $A$. 
For example, a $11.4\%$ spectral efficiency degradation occurs  for $A=1\text{ }\mathrm{mm}^2$ when $\theta$ increases from $20^o$ to $40^o$.  

Figure~\ref{fig:AvSNR_vs_A_eta_3D} depicts the average SNR as a function of the PD's effective area and quantum efficiency, $\eta$.
We observe that the value of average SNR increases as A and $\eta$ increase.
For example, for $\eta=0.3$ and $A=0.5\text{ }\mathrm{mm^2}$ the average SNR is approximately $71.86\text{ }\mathrm{dB}$, while for the same $\eta$ and $A=1.5\text{ }\mathrm{mm^2}$, it is approximately $79.12\text{ }\mathrm{dB}$. 
In addition, in the case of $A=1\text{ }\mathrm{mm^2}$ and $\eta=0.2$, the average SNR is  approximately $73.1\text{ }\mathrm{dB}$, whilst for the same $A$ and $\eta=0.8$,  it is approximately $85.15\text{ }\mathrm{dB}$. 

Figure~\ref{fig:AvSNR_vs_lambda_eta_3D} presents the average SNR as a function of the wavelength and the PD's quantum efficiency.
We observe that the OWC link achieves an average SNR in the range between $35$ dB and $88\text{ }\mathrm{dB}$.
In addition,  the PD's quantum efficiency is, as expected, proportional to the average SNR  for a fixed wavelength. Also, as analyzed in Fig.~\ref{fig:AvSNR_vs_Lambda_Jitter}, it becomes evident that a transmission window exists for wavelengths between $700\text{ }\mathrm{nm}$ and $1300\text{ }\mathrm{nm}$, where the link achieves the highest values for the average SNR.
Furthermore,  for the optimal value at  $\lambda=1100\text{ }\mathrm{nm}$, as $\eta$ increases from $0.2$ to $0.8$ the average SNR increases by  $16.5\%$.
This indicates that the selection of the operation wavelength is more critical than the quantum efficiency of the PD.

Figure~\ref{fig:C_vs_delta_Aeff} illustrates the ergodic capacity as a function of the skin thickness  and the   effective area of the PD.
 It is noticed that regardless of the values of $\delta$ and $A$, the ergodic capacity is in the range between $82$ Mbps and $140$ Mbps, which is substantially higher than the achievable rate in the baseline RF CI solution that does not exceed $1\text{ }\mathrm{Mbps}$~\cite{A:Cochlear_Implants_System_design_integration_and_evaluation,4431855}. 
As expected,   as $\delta$ increases for a fixed $A$,  the pathloss increases and therefore the capacity decreases. 
For instance,  as $\delta$ increases from $5\text{ }\mathrm{mm}$ to $9\text{ }\mathrm{mm}$, for $A=1\text{ }\mathrm{mm}^2$,  the capacity of the system decreases by a factor of $7.4\%$.
On the contrary,  the ergodic capacity also increases as $A$ increases, for a given $\delta$.  
For example, for $\delta=6\text{ }\mathrm{mm}$, the ergodic capacity increases by $88.7\%$ when $A$ varies from $0.5\text{ }\mathrm{mm}^2$ to $1.5\text{ }\mathrm{mm}^2$.
In other words, an increase on the PD's effective area can countermeasure the pathloss effect. 

In Fig.~\ref{fig:C_vs_eta_delta}, we demonstrate the joint effect of the PD's quantum efficiency and skin thickness on the OWCI performance in terms of the ergodic capacity.  
It is evident that   increasing the quantum efficiency can prevent  the capacity loss due to the skin thickness. 
For example, a PD with $\eta=0.3$ achieves an ergodic capacity of  $114.86\text{ }\mathrm{Mbps}$  for $\delta=4\text{ }\mathrm{mm}$, whilst a similar performance can be achieved for $\eta=0.7$ when the skin thickness is $9\text{ }\mathrm{mm}$. 

In Fig.~\ref{fig:C_vs_Ps_Aeff}, we illustrate the impact of the PD effective area on the corresponding spectral efficiency for different values of transmission PSD and different type of RXs, namely heterodyne and IM/DD.
As anticipated,   the spectral efficiency  is proportional to $\tilde{P}_s$ for a fixed PD effective area.   
Moreover, for a given $\tilde{P}_s$, as the PD effective area increases, the spectral efficiency also increases.  
For instance, for $\tilde{P}_s=1\text{ }\mathrm{\mu W/MHz}$ and heterodyne RX, the spectral efficiency increases by $7.9\%$ when the PD effective area changes from $0.5\text{ }\mathrm{mm}^2$ to  $1\text{ }\mathrm{mm}^2$, whereas from the same signal PSD  a PD with an effective area of  $1.5\text{ }\mathrm{mm}^2$ can achieve $4.12\%$ higher spectral efficiency compared to a PD with effective area of  $1\text{ }\mathrm{mm}^2$. 
This indicates that as the PD effective area increases, the spectral efficiency gain is not linearly increased. 
Finally, we observe that for a given signal PSD and PD effective area, the heterodyne RX outperforms the corresponding IM/DD RX. 
For example, for $\tilde{P}_s=1\text{ }\mathrm{\mu W/MHz}$ and $A=1\text{ }\mathrm{mm}^2$, the use of IM/DD instead of  heterodyne RX reduces the spectral efficiency by about $3.6\%$.
However, it is worth noting that IM/DD exhibits considerably lower complexity compared to the corresponding heterodyne RX.

\section{Conclusion} \label{S:C}
The above findings revealed that OWCIs outperforms the corresponding baseline CIs in terms of received signal quality, outage performance, spectral efficiency and channel capacity.  
In addition, they require significantly less transmit power, which renders them highly energy efficient. Notably, this difference in required power is about three orders of magnitude, since OWCI can achieve Mbps capacity levels with only few $\mu$W, contrary to conventional RF based CIs that require at least few mW to achieve this capacity. 
Furthermore, another  advantage of OWCIs in comparison with RF based CI solutions is that they operate in a non-standardized frequency region, 
where there is a significantly  large amount of unexploited bandwidth; as a result,  there is no interference from other medical implanted devices.  
Therefore, it is evident that all these factors constitute OWCI a particularly attractive alternative to the RF CI solution. 

\section*{Appendices}

\section*{Appendix A}
\section*{Proof of Theorem 1}
According to~\eqref{Eq:SNR_eq}, the instantaneous SNR is a random variable (RV) that follows the same distribution as~$h_p^2$. 
Therefore, in order to derive the average SNR of the optical link, we first need to identify the distribution of $h_p^2$. 

By assuming that the spatial intensity of $w_{\rm \delta}$  on the RX plane at distance $\delta$ from the TX and circular aperture of radius $\beta$,  the stochastic term of the channel coefficient, which represents the fraction of the collected power due to geometric spread with radial displacement $r$ from the origin of the detector, can be approximated as
\begin{align}
h_p \approx A_0\exp\left(-\frac{2r^2}{w_{\rm eq}^2}\right).
\label{Eq:hp1}
\end{align}
This is a well-known approximation that has been used in several reported contributions (see for example~\cite{A:Outage_Capacity_for_FSO_with_pointing_errors,A:BER_performance_of_FSO_link_over_strong_atm_turbulence_channels_with_pointing_errors}, and references therein).

Moreover, by assuming that the elevation and the horizontal displacement (sway) follow independent and identical Gaussian distributions and based on~\cite{A:Arnon2003Effects},we observe that the radial displacement at the RX follows a Rayleigh distribution with a probability density function (PDF) given by~\cite{B:Probability_Random_Variables_and_Stochastic_Processes}
\begin{align}
f_r\left(r\right) = \frac{r}{\sigma_s^2}\exp\left(-\frac{r^2}{2\sigma_s^2}\right), \qquad r>0
\label{Eq:fr1}
\end{align}
By combining \eqref{Eq:hp1} and \eqref{Eq:fr1}, the PDF of $h_p$ can be expressed as
\begin{align}
f_{\rm h_p}\left(x\right) = \frac{\xi}{A_0^{\xi}}x^{\xi-1}, \qquad 0 \leq x \leq A_0
\label{Eq:fhp1}
\end{align}
while, its cumulative distribution function (CDF)  can be obtained~as
\begin{align}
F_{\rm h_p}\left(x\right) \triangleq \int_0^{x} \! f_{\rm y}\left(x\right) \, \mathrm{d}y
\label{Eq:CDF_hp1}
\end{align}
which upon substituting~\eqref{Eq:fhp1} in~\eqref{Eq:CDF_hp1} and performing the integration yields
\begin{align}
F_{\rm h_p}\left(x\right) = \left\{\begin{array}{c c} \frac{1}{A_0^{\xi}}x^{\xi} & 0\leq x \leq A_0 \\ 1, & x\geq A_0\end{array}\right.
\label{CDF:hp2}
\end{align}

The CDF of $h_p^2$ can be obtained~as
\begin{align}
F_{\rm h_p^2}\left(x\right)= P\left(h_p^2 \leq x\right) = P\left(h_p \leq \sqrt{x}\right) = P\left(h_p \leq \sqrt{x}\right) = F_{\rm h_p}\left(\sqrt{x}\right)
\end{align}
which, by using~\eqref{CDF:hp2}, can be rewritten~as
\begin{align}
F_{\rm h_p^2}\left(x\right)=\left\{\begin{array}{c c} \frac{1}{A_0^{\xi}}x^{{\xi}/{2}} & 0\leq x \leq A_0^2 \\ 1, & x\geq A_0^2\end{array}\right..
\label{Eq:CDF:hp2}
\end{align}
Moreover, the PDF of $h_p^2$ can be obtained~as
\begin{align}
f_{\rm h_p^2}\left(x\right)= \frac{\mathrm{d}F_{\rm h_p^2}\left(x\right)}{\mathrm{d}x}
\label{Eq:PDF:hp2_1}
\end{align}
or, by employing~\eqref{Eq:CDF:hp2} as
\begin{align}
f_{\rm h_p^2}\left(x\right)= \frac{\xi}{2A_0^{\xi}}x^{\frac{\xi}{2}-1}.
\label{Eq:PDF:hp2}
\end{align}

Finally, the average SNR can be defined~as
\begin{align}
\tilde{\gamma}=\mathbb{E}[\gamma].
\label{Eq:AvSNR}
\end{align}
Thus, by substituting~\eqref{Eq:SNR_eq} into~\eqref{Eq:AvSNR}, and using~\eqref{Eq:PDF:hp2}, we obtain
\begin{align}
\tilde{\gamma}=\frac{R \exp\left(-\alpha(\lambda) \delta\right) \tilde{P}_s}{2 q R  P_b + 2 q   I_{\rm DC} +N_0} \frac{\xi}{2 A_0^\xi}\int_0^{A_0^2}x^{\frac{\xi}{2}}\mathrm{d}x
\end{align}
which, by performing the integration yields~\eqref{Eq:AverageSNR}, which concludes the~proof.

\section*{Appendix B}
\section*{Proof of Theorem 2}
By assuming full knowledge of the channel state information (CSI)   at both the TX and the RX, the ergodic spectral efficiency, C, (in $\mathrm{bits/channel\text{ }use}$) can be obtained~by~\cite[eq. (7.43)]{B:Advanced_OWC_Systems},~\cite[eq. (9.40)]{B:OWC_An_Emerging_Technology} and ~\cite{gao2016average,
A:ErgodicCapacityComparisonOfOpticalWirelessCommunicationsUsingAdaptiveTransmissions,
cheng2014average,
C:OWC_SystemModel_Capacity_and_Coding,
A:ImpactOfPointingErrorsOnThePerformanceOfMixedRF_FSODualHopTransmissionSystems,
A:ErgodicCapacityAnalysisOfFreeSpaceOpticalLinksWithNonzeroBoresightPointingErrors,
A:PracticalSwithcingBasedHybridFSO_RF_TransmissionAnd_its_performance_analysis,
doi:10.1117/1.OE.53.1.016107,
6987241,
6515209,
7192727,
garcia2010average}, 
namely
\begin{align}
C = \frac{1}{2} \mathbb{E}\left[\log_2\left(1 + \psi \gamma\right)\right]
\label{Eq:EC_def}
\end{align}
where 
\begin{align}
\psi = \left\{\begin{array}{c c} 1, & \text{for heterodyne RX} \\
\tfrac{e}{2\pi}, & \text{for IM/DD RX}\end{array}\right.
\end{align}
or equivalently
\begin{align}
C = \frac{1}{2} \int_{\rm 0}^{A}\log_2\left(1 + \psi \frac{R^2 \exp\left(-\alpha(\lambda) \delta\right) x^2 \tilde{P}_s}{2 q R P_b + 2 q  I_{\rm DC} + N_0}\right) f_{\rm h_p}(x)\mathrm{d}x.
\label{Eq:EC_step1}
\end{align}
By substituting~\eqref{Eq:fhp1} into~\eqref{Eq:EC_step1} and after some algebraic manipulations, equation \eqref{Eq:EC_step1} is given by
\begin{align}
C = \frac{\xi}{2 A_0^\xi\ln\left(2\right)} \mathcal{I}
\label{Eq:EC_step2}
\end{align}
where 
\begin{align}
\mathcal{I}=\int_{\rm 0}^{A_0}x^{\xi-1}\ln(1+\mathcal{B}(\lambda) x^2) \mathrm{d}x
\end{align}
or
\begin{align}
\mathcal{I}=\frac{1}{\xi}\int_{\rm 0}^{A_0}\ln(1+\mathcal{B}(\lambda) x^2) \frac{\mathrm{d}x^{\xi}}{\mathrm{d}x}\mathrm{d}x. 
\label{Eq:I_s2}
\end{align}
By applying integration by parts into~\eqref{Eq:I_s2}, we obtain
\begin{align}
\mathcal{I}=\frac{1}{\xi}A_0^\xi \ln\left(1+\mathcal{B}(\lambda) A_0^2\right) - \frac{1}{\xi}\int_0^{A_0}\frac{x^\xi}{1+\mathcal{B}(\lambda) x^2}\mathrm{d}x
\end{align}
which, by using~\cite{B:Gra_Ryz_Book}, can be expressed~as
\begin{align}
\mathcal{I}=\frac{1}{\xi}A_0^\xi \ln\left(1+\mathcal{B}(\lambda) A_0^2\right) - \frac{1}{\xi}A_0^{\xi+2} \mathcal{B}(\lambda) \Phi\left(A_0^2 \mathcal{B}(\lambda), 1, 1+\frac{\xi}{2}\right).
\label{Eq:I_s4}
\end{align}
Finally, by substituting~\eqref{Eq:I_s4} into~\eqref{Eq:EC_step2}, and after some mathematical manipulations, we obtain~\eqref{Eq:ergodic_cap}, which concludes the proof.

\section*{Appendix C}
\section*{Proof of Proposition 1}

It is recalled that the Lerch-$\Phi$ function can be expressed in integral form as 
\begin{equation}
\Phi(a, b, x) \triangleq  \frac{1}{\Gamma(b)} \int_0^\infty \frac{y^{b - 1}e^{-xy}}{1 - ae^{-y}} \mathrm{d}y.
\end{equation}
Based on this and recalling that $\Gamma(1) \triangleq 0!$,  it immediately follows that
\begin{align}
C = \frac{1}{2} \log_2\left(1+ \mathcal{B}(\lambda) A_0^2\right) - \frac{1}{2} \frac{A_0^{2} \mathcal{B}(\lambda)}{\ln(2)}  \int_0^\infty \frac{ e^{-\left(  1+\frac{\xi}{2}\right)y}}{1 + A_0^2 \mathcal{B}(\lambda)e^{-y}} \mathrm{d}y.
\label{Eq:ergodic_cap_c}
\end{align}
Notably, $1 + A_0^2 \mathcal{B}(\lambda)e^{-y}> A_0^2 \mathcal{B}(\lambda)e^{-y}$, which when $A_0^2 \mathcal{B} >> 1$ becomes $1 + A_0^2 \mathcal{B}(\lambda)e^{-y} \simeq  A_0^2 \mathcal{B}(\lambda)e^{-y}$. To this end, equation \eqref{Eq:ergodic_cap_c} cen be lower bounded as follows: 
\begin{align}
C > \frac{1}{2} \log_2\left(1+ \mathcal{B}(\lambda) A_0^2\right) - \frac{1}{2} \frac{A_0^{2} \mathcal{B}(\lambda)}{\ln(2)}  \int_0^\infty \frac{ e^{-\left(  1+\frac{\xi}{2}\right)y + y}}{ A_0^2 \mathcal{B}(\lambda)} \mathrm{d}y.
\label{Eq:ergodic_cap_c}
\end{align}
Therefore, by evaluating the above basic integral and after some algebraic manipulations one obtains \eqref{Eq:ergodic_cap_b}, which concludes the proof.

\section*{Appendix D}
\section*{Proof of Theorem 3}
We recall that the outage probability is defined as the probability that the achievable spectral efficiency,~$C$,  falls below a   predetermined threshold, $r_{\rm th}$, i.e.
\begin{align}
P_o(r_{\rm th}) = P_{\rm r}(C\leq r_{\rm th})= P_{\rm r}(\gamma\leq \gamma_{\rm th})
\label{Eq:OP_def}
\end{align}
where $r_{\rm th}$ and $\gamma_{\rm th}$ represent the data rate threshold and the SNR threshold, respectively,  whereas  
\begin{align}
\gamma_{\rm th}=\frac{2^{2 r_{\rm th}}-1}{\psi}.
\end{align}

Thus, with the aid of~\eqref{Eq:SNR_eq}, equation \eqref{Eq:OP_def} can be rewritten~as
\begin{align}
P_o(\gamma_{\rm th}) = P_{\rm r}\left(h_p^2\leq\frac{2 q R P_b + 2 q I_{\rm DC} + N_0}{R^2 \exp\left(-\alpha(\lambda)\delta\right) \tilde{P}_s} \gamma_{\rm th} \right) 
= F_{\rm h_p^2}\left(\frac{2 q R P_b + 2 q I_{\rm DC} + N_0}{R^2 \exp\left(-\alpha(\lambda)\delta\right) \tilde{P}_s} \gamma_{\rm th} \right) 
\end{align}
where $F_{\rm h_p^2}\left(\cdot\right)$ denotes  the cumulative distribution function (CDF) of the stochastic process $h_p^2$, and can  be deduced from~\eqref{Eq:fhp1}. 
As a result, \eqref{Eq:OP_def} can be rewritten in closed-form as in~\eqref{Eq:OP}, which  concludes the proof.

\section*{Funding}
Khalifa University of Science and Technology (8474000122, 8474000137).

\section*{Disclosures}
The authors declare that there are no conflicts of interest related to this article.


\bibliography{refs}

\begin{thebibliography}{10}
\newcommand{\enquote}[1]{``#1''}

\bibitem{A:The_modern_cochlear_implant}
B.~S. Wilson, \enquote{The modern cochlear implant: {A} triumph of biomedical
  engineering and the first substantial restoration of human sense using a
  medical intervention,} {\protect\JournalTitle{IEEE Pulse}} \textbf{8}, 29--32
  (2017).

\bibitem{A:Cochlear_Implants_System_design_integration_and_evaluation}
F.~G. Zeng, S.~Rebscher, W.~Harrison, X.~Sun, and H.~Feng, \enquote{Cochlear
  implants: {System} design, integration, and evaluation,}
  {\protect\JournalTitle{IEEE Rev. Biomed. Eng.}} \textbf{1}, 115--142 (2008).

\bibitem{A:MED-EL_CI}
I.~Hochmair, P.~Nopp, C.~Jolly, M.~Schmidt, H.~Sch{\"o}{\ss}er, C.~Garnham, and
  I.~Anderson, \enquote{{MED-EL} cochlear implants: {S}tate of the art and a
  glimpse into the future,} {\protect\JournalTitle{Trends Amplif.}}
  \textbf{10}, 201--219 (2006).

\bibitem{A:NFCI}
J.~F. Patrick, P.~A. Busby, and P.~J. Gibson, \enquote{The development of the
  nucleus{\textregistered} freedom cochlear implant system,}
  {\protect\JournalTitle{Trends Amplif.}} \textbf{10}, 175--200 (2006).

\bibitem{A:WPT_strategies_for_implantable_bioelectronics}
K.~Agarwal, R.~Jegadeesan, Y.~X. Guo, and N.~V. Thakor, \enquote{Wireless power
  transfer strategies for implantable bioelectronics,}
  {\protect\JournalTitle{IEEE Rev. Biomed. Eng.}} \textbf{10}, 136--161 (2017).

\bibitem{8052089}
H.~J. Kim, H.~Hirayama, S.~Kim, K.~J. Han, R.~Zhang, and J.~W. Choi,
  \enquote{Review of near-field wireless power and communication for biomedical
  applications,} {\protect\JournalTitle{IEEE Access}} \textbf{5}, 21264--21285
  (2017).

\bibitem{A:IR_neural_stimulation}
A.~C. Thompson, S.~A. Wade, N.~C. Pawsey, and P.~R. Stoddart, \enquote{Infrared
  neural stimulation: {Influence} of stimulation site spacing and repetition
  rates on heating,} {\protect\JournalTitle{IEEE Trans. Biomed. Eng.}}
  \textbf{60}, 3534--3541 (2013).

\bibitem{A:Early_History_and_Challenges_of_Implantable_electronics}
W.~H. Ko, \enquote{Early history and challenges of implantable electronics,}
  {\protect\JournalTitle{J. Emerg. Technol. Comput. Syst.}} \textbf{8}, 1--9
  (2012).

\bibitem{islam2016review}
M.~N. Islam and M.~R. Yuce, \enquote{Review of medical implant communication
  system {(MICS)} band and network,} {\protect\JournalTitle{ICT Express}}
  \textbf{2}, 188--194 (2016).

\bibitem{liu2012optical}
T.~Liu, U.~Bihr, S.~M. Anis, and M.~Ortmanns, \enquote{Optical transcutaneous
  link for low power, high data rate telemetry,} in \emph{Annual International
  Conference of the IEEE Engineering in Medicine and Biology Society (EMBC),}
  (IEEE, 2012), pp. 3535--3538.

\bibitem{pinski2002interference}
S.~L. Pinski and R.~G. Trohman, \enquote{Interference in implanted cardiac
  devices, part i,} {\protect\JournalTitle{Pacing Clin. Electrophysiol}}
  \textbf{25}, 1367--1381 (2002).

\bibitem{A:Emerging_OWC_Advances_and_Challenges}
Z.~Ghassemlooy, S.~Arnon, M.~Uysal, Z.~Xu, and J.~Cheng, \enquote{Emerging
  optical wireless communications-advances and challenges,}
  {\protect\JournalTitle{IEEE J. Sel. Areas Commun.}} \textbf{33}, 1738--1749
  (2015).

\bibitem{ghassemlooy2017visible}
Z.~Ghassemlooy, L.~N. Alves, S.~Zvanovec, and M.-A. Khalighi, \emph{Visible
  Light Communications: Theory and Applications} (CRC Press, 2017).

\bibitem{OWC_vs_RF_survey}
M.~Z. Chowdhury, M.~T. Hossan, A.~Islam, and Y.~M. Jang, \enquote{A comparative
  survey of optical wireless technologies: Architectures and applications,}
  {\protect\JournalTitle{IEEE Access}} \textbf{6}, 9819--9840 (2018).

\bibitem{abita2004transdermal}
J.~L. Abita and W.~Schneider, \enquote{Transdermal optical communications,}
  {\protect\JournalTitle{Johns Hopkins APL Tech. Dig.}} \textbf{25}, 261
  (2004).

\bibitem{ackermann2006design}
D.~M. Ackermann, B.~Smith, K.~L. Kilgore, and P.~H. Peckham, \enquote{Design of
  a high speed transcutaneous optical telemetry link,} in \emph{Engineering in
  Medicine and Biology Society, 2006. EMBS'06. 28th Annual International
  Conference of the IEEE,}  (IEEE, 2006), pp. 2932--2935.

\bibitem{ackermann2008designing}
D.~M. Ackermann~Jr, B.~Smith, X.-F. Wang, K.~L. Kilgore, and P.~H. Peckham,
  \enquote{Designing the optical interface of a transcutaneous optical
  telemetry link,} {\protect\JournalTitle{IEEE Trans. Biomed. Eng.}}
  \textbf{55}, 1365--1373 (2008).

\bibitem{gil2012feasibility}
Y.~Gil, N.~Rotter, and S.~Arnon, \enquote{Feasibility of retroreflective
  transdermal optical wireless communication,} {\protect\JournalTitle{Appl.
  Opt.}} \textbf{51}, 4232--4239 (2012).

\bibitem{liu2013system}
T.~Liu, J.~Anders, and M.~Ortmanns, \enquote{System level model for
  transcutaneous optical telemetric link,} in \emph{IEEE International
  Symposium on Circuits and Systems (ISCAS),}  (IEEE, 2013), pp. 865--868.

\bibitem{liu2014vivo}
T.~Liu, U.~Bihr, J.~Becker, J.~Anders, and M.~Ortmanns, \enquote{In vivo
  verification of a 100 mbps transcutaneous optical telemetric link,} in
  \emph{IEEE Biomedical Circuits and Systems Conference (BioCAS),}  (IEEE,
  2014), pp. 580--583.

\bibitem{liu2015bidirectional}
T.~Liu, J.~Anders, and M.~Ortmanns, \enquote{Bidirectional optical
  transcutaneous telemetric link for brain machine interface,}
  {\protect\JournalTitle{Electron. Lett.}} \textbf{51}, 1969--1971 (2015).

\bibitem{okamoto2005development}
E.~Okamoto, Y.~Yamamoto, Y.~Inoue, T.~Makino, and Y.~Mitamura,
  \enquote{Development of a bidirectional transcutaneous optical data
  transmission system for artificial hearts allowing long-distance data
  communication with low electric power consumption,} {\protect\JournalTitle{J.
  Artif. Organs}} \textbf{8}, 149--153 (2005).

\bibitem{6679681}
K.~Duncan and R.~Etienne-Cummings, \enquote{Selecting a safe power level for an
  indoor implanted uwb wireless biotelemetry link,} in \emph{IEEE Biomedical
  Circuits and Systems Conference (BioCAS),}  (2013), pp. 230--233.

\bibitem{A:GaN_based_micro_LED_arrays_on_flexible_substrates_for_optical_cochlear_implants}
C.~Go{\ss}ler, C.~Bierbrauer, R.~Moser, M.~Kunzer, K.~Holc, W.~Pletschen,
  K.~K\"ohler, J.~Wagner, M.~Schwaerzle, P.~Ruther, O.~Paul, J.~Neef,
  D.~Keppeler, G.~Hoch, T.~Moser, and U.~T. Schwarz, \enquote{{GaN}-based
  micro-{LED} arrays on flexible substrates for optical cochlear implants,}
  {\protect\JournalTitle{J. Phys. D: Appl. Phys.}} \textbf{47}, 205401 (2014).

\bibitem{kallweit2016optoacoustic}
N.~Kallweit, P.~Baumhoff, A.~Krueger, N.~Tinne, A.~Kral, T.~Ripken, and
  H.~Maier, \enquote{Optoacoustic effect is responsible for laser-induced
  cochlear responses,} {\protect\JournalTitle{Sci. Rep.}} \textbf{6}, 1--10
  (2016).

\bibitem{schultz2014optical}
M.~Schultz, P.~Baumhoff, N.~Kallweit, M.~Sato, A.~Kr{\"u}ger, T.~Ripken,
  T.~Lenarz, and A.~Kral, \enquote{Optical stimulation of the hearing and deaf
  cochlea under thermal and stress confinement condition,} in \emph{Optical
  Techniques in Neurosurgery, Neurophotonics, and Optogenetics,}
  (International Society for Optics and Photonics, 2014), p. 892816.

\bibitem{RICHTER201472}
C.-P. Richter and X.~Tan, \enquote{Photons and neurons,}
  {\protect\JournalTitle{Hear. Res.}} \textbf{311}, 72 -- 88 (2014). Annual
  Reviews.

\bibitem{duke2009combined}
A.~R. Duke, J.~M. Cayce, J.~D. Malphrus, P.~Konrad, A.~Mahadevan-Jansen, and
  E.~D. Jansen, \enquote{Combined optical and electrical stimulation of neural
  tissue in vivo,} {\protect\JournalTitle{J. Biomed. Opt.}} \textbf{14},
  060501--060501 (2009).

\bibitem{B:Gra_Ryz_Book}
I.~S. Gradshteyn and I.~M. Ryzhik, \emph{Table of Integrals, Series, and
  Products} (Academic, New York, 2000), 6th ed.

\bibitem{dorman2004design}
M.~F. Dorman and B.~S. Wilson, \enquote{The design and function of cochlear
  implants,} {\protect\JournalTitle{Am. Sci.}} \textbf{92}, 436--445 (2004).

\bibitem{A:Cochlear_Implants_A_remarkable_past_and_a_briliant_future}
B.~S. Wilson and M.~F. Dorman, \enquote{Cochlear implants: {A} remarkable past
  and a brilliant future,} {\protect\JournalTitle{Hear. Res.}} \textbf{242},
  3--21 (2008).

\bibitem{C:OWC_SystemModel_Capacity_and_Coding}
J.~Li and M.~Uysal, \enquote{Optical wireless communications: {System} model,
  capacity and coding,} in \emph{IEEE 58th Vehicular Technology Conference. VTC
  2003-Fall (IEEE Cat. No.03CH37484),}  (2003), pp. 168--172 Vol.1.

\bibitem{zedini2015multihop}
E.~Zedini and M.-S. Alouini, \enquote{Multihop relaying over im/dd fso systems
  with pointing errors,} {\protect\JournalTitle{J. Lightwave Technol.}}
  \textbf{33}, 5007--5015 (2015).

\bibitem{popoola2009bpsk}
W.~O. Popoola and Z.~Ghassemlooy, \enquote{Bpsk subcarrier intensity modulated
  free-space optical communications in atmospheric turbulence,}
  {\protect\JournalTitle{J. Lightwave Technol.}} \textbf{27}, 967--973 (2009).

\bibitem{A:Arnon2003Effects}
S.~Arnon, \enquote{Effects of atmospheric turbulence and building sway on
  optical wireless-communication systems,} {\protect\JournalTitle{Opt. Lett.}}
  \textbf{28}, 129--131 (2003).

\bibitem{B:Transdermal_optical_communications}
M.~Faria, L.~N. Alves, and P.~S. de~Brito~Andr\'e, \emph{Transdermal Optical
  Communications} (CRC Press, 2017), vol.~1, chap.~10, pp. 309--336.

\bibitem{bashkatov2011optical}
A.~N. Bashkatov, E.~A. Genina, and V.~V. Tuchin, \enquote{Optical properties of
  skin, subcutaneous, and muscle tissues: a review,} {\protect\JournalTitle{J.
  Innov. Opt. Health Sci.}} \textbf{4}, 9--38 (2011).

\bibitem{A:graaff1993Opticalpropertiesofhumandermisinvitroandinvivo}
R.~Graaff, A.~Dassel, M.~Koelink, F.~De~Mul, J.~Aarnoudse, and W.~Zijlstra,
  \enquote{Optical properties of human dermis in vitro and in vivo,}
  {\protect\JournalTitle{Appl. Opt.}} \textbf{32}, 435--447 (1993).

\bibitem{A:Chang1996Effectsofcompressiononsofttissueopticalproperties}
E.~K. Chan, B.~Sorg, D.~Protsenko, M.~O'Neil, M.~Motamedi, and A.~J. Welch,
  \enquote{Effects of compression on soft tissue optical properties,}
  {\protect\JournalTitle{IEEE J. Sel. Topics Quantum Electron.}} \textbf{2},
  943--950 (1996).

\bibitem{A:Simpson1998Near-infraredopticalproperties}
C.~R. Simpson, M.~Kohl, M.~Essenpreis, and M.~Cope, \enquote{Near-infrared
  optical properties of ex-vivo human skin and subcutaneous tissues measured
  using the monte carlo inversion technique,} {\protect\JournalTitle{Phys. Med.
  Biol.}} \textbf{43}, 2465 (1998).

\bibitem{A:Du2001Opticalpropertiesofporcineskindermis}
Y.~Du, X.~Hu, M.~Cariveau, X.~Ma, G.~Kalmus, and J.~Lu, \enquote{Optical
  properties of porcine skin dermis between 900 nm and 1500 nm,}
  {\protect\JournalTitle{Phys. Med. Biol.}} \textbf{46}, 167 (2001).

\bibitem{A:Troy2001Opticalpropertiesofhumanskin}
T.~L. Troy and S.~N. Thennadil, \enquote{Optical properties of human skin in
  the near infrared wavelength range of 1000 to 2200 nm,}
  {\protect\JournalTitle{J. Biomed. Opt.}} \textbf{6}, 167--176 (2001).

\bibitem{A:Bashkatov2005Opticalpropertiesofhumanskin}
A.~Bashkatov, E.~Genina, V.~Kochubey, and V.~Tuchin, \enquote{Optical
  properties of human skin, subcutaneous and mucous tissues in the wavelength
  range from 400 to 2000 nm,} {\protect\JournalTitle{J. Phys. D: Appl. Phys.}}
  \textbf{38}, 2543 (2005).

\bibitem{7862126}
M.~A. Esmail, H.~Fathallah, and M.~S. Alouini, \enquote{Outage probability
  analysis of {FSO} links over foggy channel,} {\protect\JournalTitle{IEEE
  Photonics J.}} \textbf{9}, 1--12 (2017).

\bibitem{A:Outage_Capacity_for_FSO_with_pointing_errors}
A.~A. Farid and S.~Hranilovic, \enquote{Outage capacity optimization for
  free-space optical links with pointing errors,} {\protect\JournalTitle{J.
  Lightwave Technol.}} \textbf{25}, 1702--1710 (2007).

\bibitem{A:OnTheCapacityOfFSOIntensityChannels}
A.~Lapidoth, S.~M. Moser, and M.~A. Wigger, \enquote{On the capacity of
  free-space optical intensity channels,} {\protect\JournalTitle{IEEE Trans.
  Inf. Theory}} \textbf{55}, 4449--4461 (2009).

\bibitem{B:Advanced_OWC_Systems}
S.~Arnon, J.~Barry, G.~Karagiannidis, R.~Schober, and M.~Uysal, eds.,
  \emph{Advanced Optical Wireless Communication Systems} (Cambridge University
  Press, New York, NY, USA, 2012), 1st ed.

\bibitem{oppermann2005uwb}
I.~Oppermann, M.~H{\"a}m{\"a}l{\"a}inen, and J.~Iinatti, \emph{{UWB:} {T}heory
  and applications} (John Wiley \& Sons, 2005).

\bibitem{D:Max3657}
Maxim Integrated Products, \emph{155 Mbps Low-Noise Transimpedance Amplifier}
  (2004). Rev. 2.

\bibitem{D:TSHG5510}
Vishay Semiconductors, \emph{High Speed Infrared Emitting Diode, 830 nm, GaAlAs
  Double Hetero} (2011). Rev. 1.2.

\bibitem{D:IR333-A}
Everlight Electronics Co Ltd, \emph{Infrared (IR) Emitter 940nm 1.2V 100mA
  7.8mW/sr @ 20mA $20^o$ Radial} (2016). Rev. 5.

\bibitem{D:LED1050E}
Thorlabs, \emph{LED1050E - Epoxy-Encased LED, 1050 nm, 2.5 mW, T-1 3/4} (2007).
  Rev. A.

\bibitem{D:LED1070E}
Thorlabs, \emph{LED1070E - Epoxy-Encased LED, 1070 nm, 7.5 mW, T-1 3/4} (2011).
  Rev. B.

\bibitem{D:LED1200E}
Thorlabs, \emph{LED1200E - Epoxy-Encased LED, 1200 nm, 2.5 mW, T-1 3/4} (2007).
  Rev. A.

\bibitem{4431855}
J.~J. Sit and R.~Sarpeshkar, \enquote{A cochlear-implant processor for encoding
  music and lowering stimulation power,} {\protect\JournalTitle{IEEE Pervasive
  Comput.}} \textbf{7}, 40--48 (2008).

\bibitem{A:BER_performance_of_FSO_link_over_strong_atm_turbulence_channels_with_pointing_errors}
H.~G. Sandalidis, T.~A. Tsiftsis, G.~K. Karagiannidis, and M.~Uysal,
  \enquote{Ber performance of {FSO} links over strong atmospheric turbulence
  channels with pointing errors,} {\protect\JournalTitle{IEEE Commun. Lett.}}
  \textbf{12}, 44--46 (2008).

\bibitem{B:Probability_Random_Variables_and_Stochastic_Processes}
A.~Papoulis and S.~Pillai, \emph{Probability, Random Variables, and Stochastic
  Processes}, McGraw-Hill series in electrical engineering: Communications and
  signal processing (Tata McGraw-Hill, 2002).

\bibitem{B:OWC_An_Emerging_Technology}
M.~Uysal, C.~Capsoni, Z.~Ghassemlooy, A.~Boucouvalas, and E.~Udvary,
  \emph{Optical Wireless Communications: {An} Emerging Technology}, Signals and
  Communication Technology (Springer International Publishing, 2016).

\bibitem{gao2016average}
J.~Gao, Y.~Zhang, M.~Cheng, Y.~Zhu, and Z.~Hu, \enquote{Average capacity of
  ground-to-train wireless optical communication links in the {non-Kolmogorov}
  and gamma--gamma distribution turbulence with pointing errors,}
  {\protect\JournalTitle{Opt. Commun.}} \textbf{358}, 147--153 (2016).

\bibitem{A:ErgodicCapacityComparisonOfOpticalWirelessCommunicationsUsingAdaptiveTransmissions}
M.~Z. Hassan, M.~J. Hossain, and J.~Cheng, \enquote{Ergodic capacity comparison
  of optical wireless communications using adaptive transmissions,}
  {\protect\JournalTitle{Opt. Express}} \textbf{21}, 20346--20362 (2013).

\bibitem{cheng2014average}
M.~Cheng, Y.~Zhang, J.~Gao, F.~Wang, and F.~Zhao, \enquote{Average capacity for
  optical wireless communication systems over exponentiated {Weibull}
  distribution {non-Kolmogorov} turbulent channels,}
  {\protect\JournalTitle{Appl. Opt.}} \textbf{53}, 4011--4017 (2014).

\bibitem{A:ImpactOfPointingErrorsOnThePerformanceOfMixedRF_FSODualHopTransmissionSystems}
I.~S. Ansari, F.~Yilmaz, and M.~S. Alouini, \enquote{Impact of pointing errors
  on the performance of mixed {RF/FSO} dual-hop transmission systems,}
  {\protect\JournalTitle{IEEE Wireless Commun. Lett.}} \textbf{2}, 351--354
  (2013).

\bibitem{A:ErgodicCapacityAnalysisOfFreeSpaceOpticalLinksWithNonzeroBoresightPointingErrors}
I.~S. Ansari, M.~S. Alouini, and J.~Cheng, \enquote{Ergodic capacity analysis
  of free-space optical links with nonzero boresight pointing errors,}
  {\protect\JournalTitle{IEEE Trans. Wireless Commun.}} \textbf{14}, 4248--4264
  (2015).

\bibitem{A:PracticalSwithcingBasedHybridFSO_RF_TransmissionAnd_its_performance_analysis}
M.~Usman, H.~C. Yang, and M.~S. Alouini, \enquote{Practical switching-based
  hybrid {FSO/RF} transmission and its performance analysis,}
  {\protect\JournalTitle{IEEE Photonics J.}} \textbf{6}, 1--13 (2014).

\bibitem{doi:10.1117/1.OE.53.1.016107}
J.-Y. Wang, J.-B. Wang, M.~Chen, N.~Huang, L.-Q. Jia, and R.~Guan,
  \enquote{Ergodic capacity and outage capacity analysis for multiple-input
  single-output free-space optical communications over composite channels,}
  {\protect\JournalTitle{Opt. Eng.}} \textbf{53}, 1--8 (2014).

\bibitem{6987241}
E.~Zedini, I.~S. Ansari, and M.~S. Alouini, \enquote{Performance analysis of
  mixed nakagami-$m$ and gamma-gamma dual-hop {FSO} transmission systems,}
  {\protect\JournalTitle{IEEE Photonics J.}} \textbf{7}, 1--20 (2015).

\bibitem{6515209}
F.~Benkhelifa, Z.~Rezki, and M.~S. Alouini, \enquote{Low snr capacity of fso
  links over gamma-gamma atmospheric turbulence channels,}
  {\protect\JournalTitle{IEEE Commun. Lett.}} \textbf{17}, 1264--1267 (2013).

\bibitem{7192727}
I.~S. Ansari, F.~Yilmaz, and M.~S. Alouini, \enquote{Performance analysis of
  free-space optical links over {M\'alaga} ($\mathcal{M}$) turbulence channels
  with pointing errors,} {\protect\JournalTitle{IEEE Trans. Wireless Commun.}}
  \textbf{15}, 91--102 (2016).

\bibitem{garcia2010average}
A.~Garcia-Zambrana, B.~Castillo-V{\'a}zquez, and C.~Castillo-V{\'a}zquez,
  \enquote{Average capacity of fso links with transmit laser selection using
  non-uniform ook signaling over exponential atmospheric turbulence channels,}
  {\protect\JournalTitle{Opt. Express}} \textbf{18}, 20445--20454 (2010).

\end{thebibliography}

\end{document}